\RequirePackage{fix-cm}
\documentclass[smallcondensed]{svjour3}     
\usepackage{graphicx}
\usepackage{newtxtext,newtxmath}
\usepackage{lineno,hyperref}
\usepackage{xspace,color,upgreek}
\usepackage[none]{hyphenat}
\usepackage{hyphenat}
\usepackage{tikz}
\usetikzlibrary{automata, arrows.meta, calc, positioning}

\journalname{Acta Informatica}

\usepackage{xargs}
\usepackage{xspace}
\usepackage{xstring}
\usepackage{boolexpr}
\usepackage[cmex10]{mathtools}
\usepackage{amsfonts}
\usepackage{mathrsfs}
\usepackage{latexsym}
\usepackage{textcomp}
\usepackage{pifont}

\newcommand{\argemp}[2]{\if&#1&\else#2\fi}

\newcommand{\argdef}[2]{\if&#1&#2\else#1\fi}

\newcommand{\argint}[3]{\if&#2&\else#1#2#3\fi}

\newcommand{\argext}[3]{\if&#1&#3\else#1\if&#3&\else#2#3\fi\fi}

\newcommandx{\mthfnt}[3][1=, 2=0]{{
	\IfStrEqCase{#1}
	{%
		{}%
		{#3}%
		{Name}%
		{%
			\IfStrEqCase{#2}
			{%
				{0}{\mathcal{#3}}%
				{1}{\mathscr{#3}}%
				{2}{\mathfrak{#3}}%
				{3}{\mathbb{#3}}%
			}
			[\ensuremath{\clubsuit}]%
		}%
		{Set}%
		{%
			\IfStrEqCase{#2}
			{%
				{0}{\mathrm{#3}}%
				{1}{\mathsf{#3}}%
				{2}{\mathbb{#3}}%
				{3}{\mathbf{#3}}%
			}
			[\ensuremath{\clubsuit}]%
		}%
		{Fun}%
		{%
			\IfStrEqCase{#2}
			{%
				{0}{\mathsf{#3}}%
				{1}{\mathrm{#3}}%
			}
			[\ensuremath{\clubsuit}]%
		}%
		{Rel}%
		{%
			\IfStrEqCase{#2}
			{%
				{0}{\mathit{#3}}%
				{1}{\mathtt{#3}}%
			}
			[\ensuremath{\clubsuit}]%
		}%
		{Sym}%
		{%
			\IfStrEqCase{#2}
			{%
				{0}{\mathtt{#3}}%
				{1}{\mathbf{#3}}%
			}
			[\ensuremath{\clubsuit}]%
		}%
		{Elm}%
		{\mathnormal{#3}}
	}
[\ensuremath{\clubsuit}]%
}}

\newcommand{\mthsub}[1]{\argemp{#1}{\ensuremath{_{\mathnormal{#1}}}}}

\newcommand{\mthsup}[1]{\argemp{#1}{\ensuremath{^{\mathnormal{#1}}}}}

\newcommandx{\mth}[5][1=, 2=0, 4=, 5=]{{\ensuremath{\mthfnt[#1][#2]{#3}\mthsub{#4}\mthsup{#5}}}}

\newcommandx{\mtharg}[6][1=, 2=0, 4=, 5=]{{\mth[#1][#2]{#3}[#4][#5]\ensuremath{\argint{(}{#6}{)}}}}

\newcommand{\mthstyname}{0}
\newcommand{\mthname}[1][]{\mth[Name][\argdef{#1}{\mthstyname}]}

\newcommand{\mthstyfun}{0}
\newcommand{\mthfun}[1][]{\mth[Fun][\argdef{#1}{\mthstyfun}]}

\newcommand{\mthstysym}{0}
\newcommand{\mthsym}[1][]{\mth[Sym][\argdef{#1}{\mthstysym}]}

\usepackage{changebar}

\def\it{\begin{itemize} }
\def\ti{\end{itemize} }
\def\en{\begin{enumerate} }
\def\ne{\end{enumerate} }

\def\exptime{\textsc{ExpTime}\xspace}
\def\exptimeC{\textsc{ExpTime-Complete}\xspace}
\def\exptimeH{\textsc{ExpTime-Hard}\xspace}

\def\ptime{\textsc{ptime}\xspace}
\def\np{\textsc{np}\xspace}

\def\tfnp{\textsc{tfnp}\xspace}

\DeclareMathOperator{\nextX}{\mathsf{X}}

\DeclareMathOperator{\until}{\mathbin{\mathsf{U}}}

\DeclareMathOperator{\always}{\mathsf{G}}

\DeclareMathOperator{\eventually}{\ensuremath{\mathsf{F}}\xspace}

\newcommand{\true}{\mathsf{true}}
\newcommand{\false}{\mathsf{false}}

\newcommand{\sat}{\ensuremath{\mathsf{sat}}\xspace}

\def\mp{{\sf mp}}
\def\max{{\sf max}}
\def\parity{{\sf parity}}

\newcommand{\LTL}{\ensuremath{\mathsf{LTL}}\xspace}

\newcommand{\AP}{\mathrm{A\!P}}

\newcommand{\tpl}[1]{\langle {#1} \rangle }
\newcommand{\tup}[1]{\overline{#1}}

\newcommand{\rst}{\upharpoonright}

\newcommand{\abs}[1]{\lvert#1\rvert}

\newcommand{\nat}{\mathbb{N}}
\def\int{\mathbb{Z}}

\newcommand{\Automaton}{\mthname{A}}
\newcommand{\Language}{\mthname{L}}

\def\avg{{\sf avg}}
\def\MC{\ensuremath{\mathcal{M}}\xspace}
\def\tot{{\sf sum}}
\def\init{\iota}
\def\src{{\sf src}}
\def\trg{{\sf trg}}
\def\IN{{\sf in}}
\def\OUT{{\sf out}}

\def\hist{{\sf Hst}} 

\newcommand{\St}{St}

\newcommand{\exec}{{\sf Exec}}

\def\Ag{{\rm Ag}}

\def\St{{\rm St}}
\def\Act{{\rm Act}}

\def\lex{\prec_{lex}}
\def\lexeq{\preceq_{lex}}
\def\LEX{{\sf Lex}}
\def\tr{\uptau}
\def\lab{\uplambda}
\def\dec{\updelta}

\def\FNE{{\sf FNE}\xspace}
\mathchardef\mhyphen="2D 

\newcommand{\pay}{\mthfun{pay}}

\def\SNE{{\sf FSNE}\xspace} 

\newcommand\SEPNE{\ensuremath{\SNE^\epsilon}\xspace}

\newcommand{\SetN}{\mathbb{N}}

\newcommand{\ex}{\mthfun{ex}}
\newcommand{\exit}{\mthsym{exit}}
\newcommand{\crash}{\mthsym{crash}}
\newcommand{\load}{\mthsym{load}}
\newcommand{\robot}{\mthsym{robot}}
\newcommand{\joint}{\mthfun{jnt}}

\newcommand{\sload}{\mthfun{l}}
\newcommand{\initial}{\mthfun{r}}

\usepackage{tikz}
\usetikzlibrary{arrows,shapes,snakes}

\usepackage{wrapfig}

\tikzstyle{every node} =
	[draw = none, fill = white, thin]
\tikzstyle{every edge} +=
	[thin]

\tikzstyle{noall} =
	[draw = none, fill = none]
\tikzstyle{nodraw} =
	[draw = none, fill = white]
\tikzstyle{nofill} =
	[draw = black, fill = none]

\tikzstyle{cnode} =
	[circle, draw = black]
\tikzstyle{snode} =
	[regular polygon, regular polygon sides = 4, draw = gray!50]
\tikzstyle{lnode} =
	[diamond, draw = gray!75]
\tikzstyle{pnode} =
	[regular polygon, regular polygon sides = 5, draw = gray]
\tikzstyle{rnode} =
[rectangle, draw = black, thin]

\newcommand{\figwarehouse}
	{
	\begin{figure}[h]
		\begin{center}
			\footnotesize
			\begin{tikzpicture}
			[node distance = 1.5cm]
			\node
			(L1)
			{$\sload_{1}$};
			\node 
			(L2)
			[below of = L1]
			{$\sload_{2}$};
			\node
			(LL1)
			[right of = L1]
			{$\initial_1$};
			\node
			(LL2)
			[right of = L2]
			{$\initial_2$};
			\node
			(ghost)
			[below of = LL1, node distance = 0.75cm]
			{};
			\node
			(joint)
			[right of = ghost]
			{$\joint$};
			\node
			(exit)
			[right of = joint]
			{$\ex$};
			\path[-]
			(L1) edge (LL1)
			(L2) edge (LL2)
			(LL1) edge (joint)
			(LL2) edge (joint)
			(joint) edge (exit)
			;
			\end{tikzpicture}
		
		\caption{\label{fig:exm:warehouse} Representation of an automated warehouse with two operating robots.}
		\end{center}
		\vspace{-3em}
	\end{figure}
	}

\modulolinenumbers[5]

\begin{document}

  \title{
  	Equilibria for Games with Combined Qualitative and Quantitative Objectives
  	\thanks{We gratefully acknowledge the financial support of
    	ERC Advanced Investigator grant 291528 at Oxford (J.\ Gutierrez, G.\ Perelli, and M.\ Wooldridge),
    	INdAM research project 2017 ``Logica e Autonomi per il Model Checking'' at Naples (A.\ Murano),
			Marie Curie Fellowship of the Istituto Nazionale di Alta Matematica (S.\ Rubin),
			the EPSRC Centre for Doctoral Training in Autonomous Intelligent Machines and Systems EP/L015897/1 and the Ian Palmer Memorial Scholarship (T.\ Steeples).
		A preliminary version of this work appeared in~\cite{GMPRW:IJCAI17}.
  	}
  }


\author{Julian Gutierrez \and 
Aniello Murano \and
Giuseppe Perelli \and 
Sasha Rubin \and 
Thomas Steeples \and
Michael Wooldridge}


\institute{
	J. Gutierrez \at Monash University\\
	\email{julian.gutierrez@monash.edu} \\
  \and
	A. Murano \at University of Naples ``Federico II''\\
	\email{murano@na.infn.it} \\
	\and
	G. Perelli \at Sapienza University of Rome \\
	\email{perelli@diag.uniroma1.it}\\        	
 	\and
	S. Rubin \at University of Sydney\\
	\email{sasha.rubin@sydney.edu.au} \\
	\and
	T. Steeples \at University of Oxford \\
	\email{thomas.steeples@cs.ox.ac.uk}\\
	\and
	M. Wooldridge \at University of Oxford\\
	\email{michael.wooldridge@cs.ox.ac.uk}
}

\date{Received: date / Accepted: date}

\maketitle

\begin{abstract}
	The overall aim of our research is to develop techniques to reason about the 
	equilibrium properties of multi-agent systems.
	We model multi-agent systems as concurrent games, in which each player is a
  process that is assumed to act independently and strategically in
  pursuit of personal preferences.
  In this article, we study these games in the context of finite-memory 
  strategies, and we assume players' preferences are defined by a qualitative 
  and a quantitative objective, which are related by a lexicographic order: a 
  player first prefers to satisfy its qualitative objective (given as a
  formula of Linear Temporal Logic) and then prefers to minimise costs
  (given by a mean-payoff function).
  Our main result is that deciding the existence of a strict $\epsilon$ Nash 
  equilibrium in such games is 2\textsc{ExpTime}-complete (and hence 
  decidable), even if players' deviations are implemented as infinite-memory 
  strategies.
\end{abstract}


\section{Introduction}

The last twenty years have seen considerable research directed at the
use of game theoretic techniques in the analysis and verification of
multi-agent systems~\cite{shoham:2008a}. From this standpoint,
agents/processes in a multi-agent system can be understood as players
in a game played on a directed graph (a transition system), acting
strategically and independently in pursuit of their preferences. In
this setting, possible behaviours of agents correspond to the
strategies of players. One important strand of work in this tradition
has been the development of techniques for reasoning about what
properties players (or coalitions of players) can bring about ({\em
  i.e.}, whether they have ``winning strategies'' for certain
conditions)~\cite{AHK02}. Recently, attention has begun to shift from
the analysis of strategic ability to the analysis of the
\emph{equilibrium properties} of such systems. A typical question in
this setting is whether a particular temporal property will hold under
the assumption that players select strategies that collectively form a
Nash equilibrium~\cite{osborne:94a}.

A fundamental question in this work is how the preferences of agents
are represented. One widely-adopted answer to this question is to
associate with each player a \emph{qualitative} goal (objective),
usually given either by a temporal logic formula or else by a winning
(acceptance) condition, such as reachability, safety, B{\" u}chi,
Linear Temporal Logic, etc.~\cite{PR89,GHW15,BBMU15,gutierrez:2017a}.
This approach is closely related to the verification of finite-state
systems, and the model checking paradigm in
particular~\cite{clarke:2000a}. However, the preference structures
that are induced in this way have a rather simple (dichotomous)
structure: a player is simply either satisfied or unsatisfied; no
distinction is made between outcomes that satisfy the player's
objective, nor is any made between outcomes that do not satisfy the
objective. This limits the applicability of such representations for
modelling many situations of interest.  An alternative setting is
given by games where, instead of having a qualitative objective,
players have \emph{quantitative} goals --- for instance, to minimize a
given cost, or to maximise some reward~\cite{EM79,UW11}. Yet a third
possibility, also the focus in this paper, is to use preference models that \emph{combine}
qualitative and quantitative
objectives~\cite{CHJ05,BCHJ09,wooldridge:2013a}. 

We consider goals given by a lexicographic
order, where a player's primary goal is to satisfy its qualitative
objective (given by a formula of Linear Temporal Logic,
\LTL~\cite{demri:2017a}), and a player's secondary goal is to minimise its
costs (where costs are given by a quantitative mean-payoff
objective). This approach has several advantages. The qualitative
objective can be used to express desirable properties on the states of
the system, as is standard practice in the specification and
verification of reactive
systems~\cite{emerson:90a,clarke:2000a,demri:2017a}. For instance,
\LTL\ formulae, which we use to express the
qualitative objective of agents, can be used to specify in a natural
way that an agent prefers not to enter a given set of states (formally
expressed as a safety property) or that an agent desires to eventually
visit a given state of the system (for instance to model the termination
of a task)~\cite{demri:2017a}. In these cases, we can simply
understand agents as ``rational processes'' within a reactive
system. The quantitative objective can be used to restrict the
behaviour of such agents to those that are locally optimal for each
agent. Thus, not only do we want an agent to accomplish its goal, but
to do so \emph{as efficiently as possible}, that is, such that the
cost of performing the task is kept to a minimum.
This latter type of preference is very naturally captured by
\emph{mean-payoff} specifications, even in cases where infinite
behaviour is considered.  Moreover, the combination of qualitative and
quantitative objectives is natural for situations in which agents aim
to satisfy some goal while minimising costs. For example, consider a
robot whose task is to deliver packages around a factory environment:
the primary goal of the robot is to deliver the packages (a
qualitative objective readily expressible in \LTL), while the
secondary goal is to minimise fuel consumption when achieving this
task (a quantitative objective that can be naturally expressed with a
mean-payoff function). Scenarios like these are ubiquitous in embedded
and cyber-physical systems~\cite{alur:2015a}.

The main solution concept we use in this paper is \emph{strict $\epsilon$ Nash equilibrium}~\cite{shoham:2008a}.
The use of Nash equilibrium --- where no player in the game can unilaterally change their strategy and be better off as a consequence --- is readily justified by the fact that this is the best-known and most widely-used solution concept for non-cooperative games.
The use of strict $\epsilon$ Nash equilibrium is less common.
Informally, by strict we mean that any possible unilateral deviation of a player
results in an outcome that is strictly worse for that player.
As such, this solution concept is more stable than ``ordinary'' Nash equilibrium since each player has less incentive to change strategy.
Thus, from a stability point of view, strict Nash equilibrium is a desirable feature.~\footnote{We remark that strict Nash equilibria appear naturally in the study of evolutionary game-theory~\cite{Smith:82}. Indeed, every strict Nash equilibrium in a symmetric game is evolutionary stable.}
However, as expected, a game may have more Nash equilibria than strict Nash equilibria (and every strict Nash equilibrium is already an ordinary Nash equilibrium).
In contrast, allowing for an $\epsilon$ Nash equilibrium, with $\epsilon>0$, may lead to games with more equilibria.
Informally, in an $\epsilon$ Nash equilibrium no player can unilaterally change their strategy and achieve a payoff that is at least as good as $\epsilon$ more than the one already obtained in the Nash equilibrium.
As a consequence, every ordinary Nash equilibrium is also an $\epsilon$ Nash equilibrium (for all $\epsilon>0$), but the converse may fail.
Then, while the strict variant of Nash equilibria can decrease the number of equilibria in a game, the $\epsilon$ variant may increase it.

\paragraph{Contributions} 
We consider games in which preferences are defined by a lexicographic order of goals given by an \LTL\ formula (the primary goal) and a mean-payoff condition (the secondary goal). We prove that deciding the existence of a finite-state strict $\epsilon$ Nash equilibrium is in 2\textsc{ExpTime}.\footnote{The transition systems on which these games are played are sometimes called ``concurrent game structures'' and sometimes ``arenas''.
Note that ``concurrent'' in this context simply means that players move at the same time, {\em i.e.}, synchronously.} 
This result subsumes the case for ordinary Nash equilibrium in
two-player zero-sum games with \LTL\ goals, which is
known to be 2\textsc{ExpTime}-hard. Thus the above problem is 
2\textsc{ExpTime}-complete.
To obtain this result we introduce a reduction to a similar (though
doubly exponentially larger) game where, instead of goals given by
\LTL\ formulae, goals are given by a parity acceptance
condition, and we show how to solve such games in \np.

Our results also show how to solve the rational synthesis
problem~\cite{FKL10} and the rational verification
problem~\cite{WGHMPT16,gutierrez:2017a} within the same complexity
class.  These problems concern establishing which
properties ({\em e.g.}, temporal, $\omega$-regular, etc.)  hold in a game, under the assumption that players in the game choose
strategies in equilibrium.  More specifically, the questions that we
ask are as follows. Given a game as described before, in which each
player has a qualitative goal given by an \LTL\ 
formula, and a quantitative goal given by a mean-payoff condition, we
ask whether a given \LTL\ formula, say~$\phi$, is
satisfied on some/all (strict $\epsilon$) Nash equilibrium/equilibria, if any, of the
game. Since~$\phi$ is not a goal of any of the players, it can be seen
as a property that the ``designer'' of the game wants to see satisfied
assuming rational behaviour of the players/agents in the game/system.

The remainder of the paper is structured as
follows:
\begin{itemize}
\item Section~\ref{secn:concgames} introduces our formal framework and
  defines the games we study throughout the paper.
\item Section~\ref{sec:decision procedures} defines the solution
  concepts underlying our main results, and presents the
  constructions, reductions, and algorithms to solve the main problems
  considered in the paper.
\item Section~\ref{secn:relwork} discusses related work.
\item Finally, Section~\ref{secn:conc} presents a number of
concluding remarks and directions for potential future work.
\end{itemize}

\section{Game Structures}
\label{secn:concgames}

In this section we introduce our game model. We use multi-player games played on finite directed graphs (transition systems), rather than games in extensive-form or normal-form~\cite{osborne:94a}.
Agents move synchronously (which includes the special sequential case), play deterministic (rather than randomised) and finite-state (instead of simply memoryless or infinite-memory) strategies, in pursuit of their individual preferences, which are given as a lexicographic combination of a qualitative temporal-logic property (intuitively, a goal/objective) and a quantitative long-term average of the rewards of its actions.

We fix some notation. If $X$ is a set, then $X^\omega$ is the set of
all infinite sequences over $X$.
If $\alpha$ is a sequence and $n \in \nat$ (the set of non-negative integers) then $\alpha_n$ represents the $(n+1)$st element of $\alpha$. 
If $X$ and $Y$ are sets, then $X^Y$
is the set of all functions $\upalpha:Y \to X$.  We will often use 
Greek letters $\upalpha, \upbeta,\upkappa,\cdots$ to name
functions.  Also, we will use tuple notation: we write
$\upalpha_y \in X$ instead of $\upalpha(y)$. We write $\AP$ for a 
finite set of \emph{atomic propositions} (or \emph{atoms} for short). 

We now define the framework of (propositional) Linear Temporal Logic
(\LTL), which we use extensively in what follows. Our presentation is
complete but is not intended as an introduction to this well-known
language: see, {em e.g.},~\cite{demri:2017a} for a detailed overview.
\paragraph{Linear-Temporal Logic (\LTL)}
	The \emph{formulae of \LTL (over $\AP$)} are generated by the following grammar:
	$$\varphi ::= p \mid \varphi \wedge \varphi \mid \neg \varphi \mid  \nextX  \varphi \mid  \varphi \until \varphi$$
	where $p \in \AP$. We use the standard classical logic
  abbreviations, {\em e.g.}, $\false := p \wedge \neg p$, as well as
  those for \LTL, {\em e.g.}, $\eventually \varphi := \true \until \varphi$ and 
  $\always \varphi := \neg \eventually \neg \varphi$.
	
	Formulae of \LTL\ are interpreted over infinite words $\alpha \in (2^{AP})^\omega$. Define the \emph{satisfaction} relation 
	$\models$ as follows:
	\begin{itemize}
		\item
			$(\alpha,n) \models p$ iff $p \in \alpha_n$; 
		
		\item
			$(\alpha,n) \models \varphi_1 \wedge \varphi_2$ iff $(\alpha,n) \models \varphi_i$ for $i = 1,2$; 
		
		\item
			$(\alpha,n) \models \neg \varphi$ iff it is not the case that $(\alpha,n) \models \varphi$; 
		
		\item
			$(\alpha,n) \models \nextX \varphi$ iff $(\alpha,n+1) \models \varphi$;
		
		\item
			$(\alpha,n) \models \varphi_1 \until \varphi_2$ iff there exists $j \geq n$ such that $(\alpha,j) \models \varphi_2$ and for all $n \leq i < j $, $(\alpha,i) \models \varphi_1$.  
	\end{itemize} 
	
Finally, define $\alpha \models \varphi$ if $(\alpha,0) \models \varphi$.
The \emph{size} of a formula is simply the number of operators within it.

\paragraph{Arenas and Lexicographic Games}	\label{def:games}
	An \emph{arena} is a tuple
	\[A = \tpl{\Ag,  \Act, \St, \iota, \tr}\]
	where $\Ag$, $\Act$, and $\St$ are finite non-empty sets of
  \emph{agents}, \emph{actions}, and
  \emph{states}, respectively;	$\iota \in \St$ is the \emph{initial
    state}; and $\tr : \St \times \Act^{\Ag} \rightarrow \St$ is a
  \emph{transition function} mapping each pair consisting of a state
  and an action for each agent, namely a \emph{decision} $\dec \in
  \Act^{\Ag}$, to a successor state. 
	
	A \emph{$\LEX(\LTL,\mp)$ game} is a tuple 
	\[ G = \tpl{A, (\upkappa_{a})_{a \in \Ag}, \AP, \lab, (\upgamma_{a})_{a \in \Ag} }\] where $A$ is an arena; $\upkappa_{a}: \St \to \int$ is a \emph{weight function} for agent $a \in \Ag$ associating an integer \emph{weight} to each state; $\AP$ is a finite set of \emph{atomic propositions} ; $\lab: \St \to 2^{\AP}$ is a \emph{labelling function} assigning a subset of atomic propositions to every state of the arena;
	and
	$\upgamma_{a}$ is an \LTL formula over $\AP$, called the \LTL \emph{goal} associated with agent $a$.

In the following, we introduce some basic notions related to games.

\paragraph{Executions}
	A \emph{path} $\pi = s_0 \dec_0 s_1 \dec_1 \cdots$ is an infinite sequence over $\St \times \Act^{\Ag}$ such that $\tr (s_i,\dec_i) = s_{i+1}$ for all $i$.
	In particular, $\dec_i(a)$ is the action of agent $a$ in step $i$. 

	A path $\pi$ induces:
	\begin{enumerate}
		\item
			the sequence $\lambda(\pi) = \lambda(s_0) \lambda(s_1) \cdots$ of sets of atoms, and
		
		\item
			for each agent $a$, the sequence $\upkappa_a(\pi) = \upkappa_a(s_0) \upkappa_a(s_1) \cdots$ of weights.
	\end{enumerate}
	
	An \emph{execution} is a path with $s_0 = \iota$. Let $\exec$ denote the set of all executions. 

\paragraph{Qualitative Goals}
In this work, qualitative goals are represented by \LTL formulas. If $\gamma$ is an \LTL formula and $\pi$ is an execution, we say that \emph{$\pi$ satisfies $\gamma$}, and write $\pi \models \gamma$, if $\lambda(\pi) \models \gamma$. 

For an execution $\pi \in \exec$ and an \LTL goal $\upgamma_a$, 
define 
$$
\sat_a(\pi) = \begin{cases}
\top & \text{ if } \pi \models \upgamma_a\\
\bot & \text{ otherwise.}
\end{cases} 
$$

\paragraph{Quantitative Goals}

	For a sequence $\alpha \in \mathbb{R}^\omega$, let $\mp(\alpha)$ be
	the \emph{mean-payoff} of $\alpha$, that is, 
	\[ \mp(\alpha) = \liminf_{n \to \infty} \avg_n(\alpha)	 
	\]
	where, for $n \in \mathbb{N}$, we define
\[\avg_n(\alpha) = \frac{1}{n}\sum_{j=0}^{n-1} \alpha_j.\] 
	This definition naturally extends to executions, {\em i.e.}, define $\mp_a(\pi) = \mp(\upkappa_a(\pi))$.

\paragraph{Lexicographic Payoffs}
	Let $\Omega = \{\bot,\top\} \times \mathbb{R}$ denote the set of \emph{payoffs}, and define the \emph{payoff function} for agent $a$ to be $\pay_{a}: \exec \to \Omega$ by $\pay_a(\pi) = (\sat_a(\pi),\mp_a(\pi))$.
	Each agent is trying to maximise its payoff. In other words, agent $a$'s primary goal is to \emph{satisfy} its \LTL formula $\gamma_{a}$, and its secondary goal is to \emph{maximise} its $\mp$-reward $\mp_a(\pi)$. 
	Formally, we define the \emph{lexicographic ordering} on the set $\Omega$ of payoffs:
	$(x,y) \lex (x',y')$ iff, either ($x = \bot$ and $x'= \top$) or ($x = x'$ and $y < y'$).
	Note that this ordering is total.
	
	\begin{remark}
		We consider weights as rewards to be maximised. However, one may be concerned with games where the agents have costs they want to minimise. One immediate thought is to take such a cost-game, replace all the weights by their negation and consider the resulting maximisation problem. However, the resulting game will not be strategically identical to the original cost game, since for an arbitrary execution $\pi$, we do not have $-\mp(\upkappa_a(\pi)) = \mp(-\upkappa_a(\pi))$ in general. 
		
		One easy way to see this is to consider the arena, $A$, with two states, $s^1, s^2$, (the number of agents and their available actions are not relevant for the most part) with a transition function such that the players can either stay in the same state, or move to the other state. Moreover, we set $s^1$ to be the start state. Thus, the arena looks like this~\footnote{Note that the arena is similar to the one of \cite[Fig. 3 in Lemma 7]{VCHRR15} to prove that Multi mean-payoff games require in general infinite memory strategies to be played optimally.}:

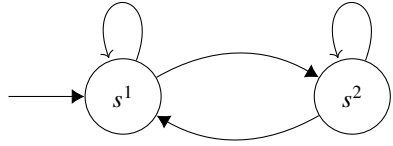
		\begin{center}
			\begin{tikzpicture}[state/.style={circle, draw, minimum size=1cm}]
				\node[state] (s^1) {$s^1$};
				\node[] (start) [left =1cm of s^1] {};
				\node[state] (s^2) [right =2cm of s^1] {$s^2$};
	
				\draw [-{Latex[width=2mm]}]
					(start) edge (s^1)
					(s^1) edge[out=30, in=150] (s^2)
					(s^1) edge[out=70, in=110, loop] (s^1)
					(s^2) edge[out=210, in=330] (s^1)
					(s^2) edge[out=70, in=110, loop] (s^2)
				;
			\end{tikzpicture}
		\end{center}

		Now, for a given player, $a$, set $\upkappa_a(s^1) = 0$ and $\upkappa_a(s^2) =1$. Now define two sequences, $a_n, b_n$, with,
		\begin{align*}
			a_0 &= 0, \\
			b_0 &= 3, \\
			a_{n+1} &= 3b_n + 2, \\
			b_{n+1} &= 3a_{n+1} + 2,
		\end{align*}
		for all $n \geq 0$. It is easy to see that we have both $a_n < b_n$, as well as $b_n < a_{n+1}$ for all $n \geq 0$. With this, we define an execution $\pi$ such that $\pi[k] = s^1$ if there exists some $n$ such that $a_n < k < b_n$ and $\pi[k] = s^2$ otherwise. Intuitively, $\pi$ bounces between $s^1$ and $s^2$, spending three times as long on each state as it did on the previous state. 

		It is easy to verify that $\avg_{a_n}(\pi) \leq 0.25$ and that $\avg_{b_n}(\pi) \geq 0.75$ for all $n$. Thus, we have,
		\begin{align*}
			-\mp(\upkappa_a(\pi)) \geq -0.25, \\
			\mp(-\upkappa_a(\pi)) \leq -0.75.
		\end{align*}

		Thus, to reason about cost-games, we cannot simply negate the weights and consider the resulting maximisation problem.

		Whilst this may seem unsatisfying, there are two ways out here. Firstly, for finite state strategies, these values \emph{do} coincide, as these induce ultimately periodic executions, whose sequence of weights will have a well-defined limit-average. Secondly, whilst the negative weights may not directly encode the original game, they do generally reflect the strategic nature of it - the resulting game is not \emph{completely} disparate from the game it was based on. Additionally, checking a given strategy profile to see if it is a Nash equilibrium is generally much easier than synthesising one in the first place. As such, we can consider the maximisation game and then try and translate our understanding of it to the original cost game. 
	\end{remark}

\paragraph{Strategies}
	A \emph{history} is a finite, possibly empty, sequence $\dec_0 \dec_1 \cdots \dec_{n-1}$ of decisions. The set of all histories is denoted $\hist$. A \emph{strategy}, $\upsigma$, for agent $a \in \Ag$ is a function $\hist \to \Act$. We emphasize that strategies map finite sequences of decisions (not states) to actions. This differs from the conventional definition in the literature. However, with this alternative definition, the Nash equilibria of the game are invariant under bisimulation \cite{DBLP:journals/lmcs/0001HPW19} -- with the conventional definition, they are not. Thus, we use our model so that our algorithm produces equilibria that have the useful property of being invariant under bisimulation.
	
	A \emph{strategy profile} is a function $\vec{\upsigma}: \Ag \to (\hist \to \Act)$. A strategy profile $\vec{\upsigma}$ induces a unique execution $\pi_{\vec{\upsigma}}$, {\em i.e.}, the execution $\pi_{\vec{\upsigma}} = s_0 \dec_0 s_1 \dec_1 \cdots$ such that $s_0 = \iota$ and $\dec_i(a) = \vec{\upsigma}(a)(\dec_0 \dec_1 \cdots \dec_{i-1})$ for $i \geq 0$. 
	
	Let $a \in \Ag$ be a player, $\vec{\upsigma}$ a strategy profile, and $\upsigma_a^{\prime}$ be an additional strategy for player $a$ - for convenience, we introduce two associated functions, $\vec{\upsigma}_{-a} : \Ag \setminus \{a\} \to (\hist \to \Act)$ and $(\vec{\upsigma}_{-a}, \upsigma_a^{\prime}) : \Ag \to (\hist \to \Act)$. Semantically, we define these as follows: $\vec{\upsigma}_{-a}(b) = \vec{\upsigma}$ for all $b \in \Ag \setminus \{a\}$, and $(\vec{\upsigma}_{-a}, \upsigma_a^{\prime})(a) = \upsigma_a^{\prime}$ with $(\vec{\upsigma}_{-a}, \upsigma_a^{\prime})(b) = \vec{\upsigma}(b)$ for all $b \in \Ag$ with $b \neq a$.

\paragraph{Finite-state strategies}
	A strategy $\upsigma$ is \emph{finite-state} if it is generated by an automaton $M$ with input alphabet $\Sigma = \Act^{\Ag}$ and with output function $\lambda:Q \to \Act$. That is, on input $h \in \hist$, the automaton $M$ reaches a state $q_h$ such that $\lambda(q_h) = \upsigma(h)$. A strategy profile, $\vec{\upsigma}$ is \emph{finite-state} if every strategy $\vec{\upsigma}(a)$ is finite-state. Observe that in this case, the unique execution $\pi_{\vec{\upsigma}}$ is ultimately periodic. 

\begin{remark}
	There are several reasons for considering finite state strategies, rather than strategies with unbounded memory.
	First, from a modelling perspective (especially within the AI and multi-agent systems communities), this is a desirable representation as it can be used to synthesise models for individual agents in a system.
	Second, from a more theoretical and computational standpoint, in games with quantitative objectives, finite state strategies can render decidable settings that would be undecidable otherwise~\cite{UW11}.
	Finally, the use of finite state machines to capture strategies in games played over an infinite number of rounds is standard in the game theory literature~\cite{binmore:92a}.
\end{remark}

	\paragraph{Strict $\epsilon$ Nash-equilibria}
	The solution concept we work with is the \emph{strict $\epsilon$ Nash-equilibrium}.
	This is a natural refinement of $\epsilon$ Nash-equilibrium~\cite{shoham:2008a}, and moreover includes strict
	Nash equilibrium as a special case.
	For $\epsilon \geq 0$ and $(x,y) \in \Omega$, let $(x,y) + \epsilon$ denote
	$(x,y+\epsilon) \in \Omega$.
	A strategy profile $\vec{\upsigma}$ is a \emph{strict $\epsilon$ Nash-equilibrium} if for every agent $a \in \Ag$, and every strategy $\upsigma_a^\prime \neq \vec{\upsigma}_a$ for $a$, we have that
	$\pay_a(\pi_{(\vec{\upsigma}_{-a},\upsigma_a^\prime)}) \lex \pay_a(\pi_{\vec{\upsigma}}) + \epsilon$.
	If $\epsilon = 0$ then we call this a \emph{strict Nash equilibrium}.
	We remark that an (ordinary) Nash equilibrium uses $\lexeq$ instead of
	$\lex$.
	By $\SEPNE(G)$ we denote the set of Finite-state Strict $\epsilon$ Nash Equilibria in $G$.
	We emphasise that, in the definition of a finite-state strict $\epsilon$ Nash-equilibrium $\vec{\upsigma}$, the deviating strategies $\vec{\upsigma}'(a)$ need not be finite-state.
	Intuitively, this captures worst-case behaviour of the deviators.

\paragraph{Decision Problems}
	The central decision problem of this work, called Rational Synthesis or Rational Verification~\cite{FKL10,GHW15,WGHMPT16,KPV16}, asks if there exists a \SEPNE so that the induced play $\pi_{\vec{\upsigma}}$ satisfies a given \LTL condition $\Phi$. Note that in case $\Phi = \top$ this amounts to the deciding the existence of a \SEPNE.

	Formally, for a rational $\epsilon \geq 0$ we consider the following decision problems for the class of $\LEX(\LTL,\mp)$-games:

\begin{itemize}
	\item \emph{$\SEPNE$-emptiness}. \\
	\emph{Given}: game $G$\\
	\emph{Question}: Is it the case that $\SEPNE(G) \neq \emptyset$?
	\item \emph{$\SEPNE$-existence}. \\
	\emph{Given}: Game $G$ and  \LTL formula $\Phi$.\\
	\emph{Question}: Does there exist a $\vec{\upsigma} \in \SEPNE(G)$ such that 
	$\pi_{\vec{\upsigma}} \models \Phi$?
\end{itemize}

To help the reader better understand our model and its applicability,
we present an example.

\begin{example}
	In an automated warehouse, $n$ robots move around to load items and bring them to the exit. Their objective is to load and unload as efficiently as possible---without crashing into each other.

	Assume that the warehouse is represented by a directed edge-labeled graph $\mathcal{G} = (V, E)$ where $E:V \times D \to V$ for some finite set of \emph{directions} $D$. For instance, if the warehouse is a grid then we may take $D = \{north,south,east,west\}$ and an agent in position $v$ executing action $d \in D$ will move to position $E(v,d)$. In addition, we are given particular vertices: for each robot $a \in \Ag$, a vertex $\initial_a$ representing its initial position, $\ex \in V$ representing the exit and $\sload_{1}, \ldots, \sload_{n} \in V$ representing the loading points for the agents. We assume loading points are different from each other and from the exit.
	
	We can model the setting by means of the arena
	\[A = \tpl{\Ag,  \Act, \St, \iota, \tr}\]
	where 
	\begin{itemize} 
		\item the agents $\Ag =\{1, \ldots, n \}$ are the robots moving around the warehouse; 
		\item $\Act = D$ are the actions that each robot can take;  
		\item $\St = V^{n}$ is the set of states of the system, where the $a$-th component of the tuple denotes the position of robot $a$;	
		\item $\iota = (\initial_{1}, \ldots, \initial_{n})$ is the initial state, denoting the initial position of the robots; 
		\item and $\tr : \St \times \Act^{\Ag} \rightarrow \St$ maps $(s,\delta)$ to the state whose $a$th component is $E(s_a,\delta(a))$.
	\end{itemize}
	
	To capture robot objectives, we define the {$\LEX(\LTL,\mp)$ game}
	
	\[ G = \tpl{A, (\upkappa_{a})_{a \in \Ag}, \AP, \lab, (\upgamma_{a})_{a \in \Ag} }\]
	
	where $A$ is the arena described above and 
	
	\begin{itemize}
		\item
			$\kappa_{a}(v) = 1$ if $v = \sload_a$, and $0$ otherwise (the weight function rewards the agent anytime it hits the loading point);
			
		\item
			$\AP = \{\exit_{1}, \ldots, \exit_{n}, \load_{1}, \ldots, \load_{n}, \crash_{1}, \ldots,\crash_n \}$, denoting the agent $a$ is at the exit ($\exit_{a}$), or the loading vertex ($\load_{a}$), or crashing with another agent ($\crash_a$);
				
		\item
			For every state $s \in V$ and agent $a \in \Ag$, we have that 
			\begin{itemize} 
			\item $\exit_{a} \in \lambda(s)$ iff $s_{a} = \ex$, 
			\item $\load_{a} \in \lambda(s)$ iff $s_{a} = \sload_{a}$, 
			\item $\crash_a \in \lambda(s)$ iff $s_a = s_b$ for some $b \neq a$ (i.e., another agent occupies the same position as agent $a$);
			\end{itemize}
			
		\item
			for every agent $a$, $\gamma_{a} = \always (\neg \crash_a) \wedge \always (\load_{a} \to \nextX (\neg \load_{a} \until \exit_{a}) )$, i.e., the agent is trying to ensure that it never crashes with another agent and that it visits the exit point between 
			every two occurrences of a visit to the loading point (this captures that the agent successfully loads and unloads an item).
			
	\end{itemize}

	Intuitively, an agent's primary objective (the \LTL formula) is to never crash and never load two items in a row before reaching the exit.
	Furthermore, an agent's secondary objective (the mean-payoff value) is to ship items as fast as possible, in order to maximise, in the limit-average, the number of times it reaches its loading point.
	
	\figwarehouse
	What might an equilibrium for this game look like? Observe that agents have essentially two possible behaviours to satisfy the primary goal: \emph{idle} and \emph{cycle}. The idle allows the robots a finite number of trips from the loading point to the exit point before keeping them forever away from the loading point, whereas the cycle makes them move back and forth between the loading point and the exit forever. In both cases, they additionally have to avoid crashing with each other. The secondary goal makes every agent to prefer the cycle behaviour. Indeed, the only way to get a strictly positive mean-payoff reward is to reach the loading point infinitely many times (with bounded delay among two consecutive times).
	
	Consider a warehouse with two operating robots as represented in Figure~\ref{fig:exm:warehouse}, and assume that the initial positions of $\robot_1$ and $\robot_2$ are $\initial_1$ and $\initial_2$, respectively.
	Moreover, consider the infinite path $\pi = u \cdot (u')^{\omega}$ with
	
	$$u = (\initial_1, \initial_2), (\sload_1, \sload_2), (\initial_1, \initial_2), (\joint, \initial_2), (\ex, \initial_2)
	$$
	\noindent and
	$$
	u' = (\joint, \initial_2) (\initial_1, \joint), (\sload_1, \ex), (\initial_1, \joint), (\joint, \initial_2), (\ex, \sload_2)
	$$
	
	\noindent where each pair represents the position of the robot at any step.
	
	The path satisfies both $\gamma_1$ and $\gamma_2$.
	Moreover, every robot cycles from the loading point to exit and back in six steps.
	Therefore, their mean-payoff value on $\pi$ is given by $\frac{1}{6}$.
	Observe that it would not be possible to improve such payoff without violating the primary objective.
	Thus, $\pi$ achieves optimal value for both robots and it is therefore a $\epsilon$-Nash Equilibrium, for every $\epsilon \geq 0$ (in particular, it is a Nash Equilibrium for $\epsilon = 0$).
	It is also strict, as any other path of $\robot_a$ satisfying $\gamma_{a}$ decreases the mean-payoff value.
	
	Observe that for a large enough value of $\epsilon$, every path that satisfies both $\gamma_1$ and $\gamma_2$ is a strict $\epsilon$ Nash Equilibrium.
	In particular, for $\epsilon = \frac{1}{6}$, the idle strategies for the robots producing the outcome $(\initial_1, \initial_2)^\omega$ is a strict $\epsilon$ Nash Equilibrium, as the primary goals are all achieved, and the secondary goal cannot be improved by strictly more than $\frac{1}{6}$.
	This is obviously not desirable from the designers point of view, as we might require that none of the agents remains idle.
	Thus, we might require an equilibrium to satisfy the \LTL formula $\Phi = \bigwedge_a \always \eventually \load_{a}$ which states that each agent is loading a fresh item infinitely often.
	The existence of such a solution can be checked by solving the \SEPNE-existence problem.
\end{example}

\section{FSNE$^{\epsilon}$-Existence and FSNE$^{\epsilon}$-Emptiness are 2\textsc{ExpTime}-complete} 
\label{sec:decision procedures}

In this section we establish our main technical result, {\em i.e.}, that \SEPNE-emptiness is in $2$\exptime. We then show that \SEPNE-existence is $2$\exptime-complete - we show membership by reducing to the \SEPNE-emptiness problem and show hardness using a reduction from \LTL games. We remark that some of the assumptions in these results are quite general (i.e., using \LTL as a specification language for qualitative properties, mean-payoff for quantitative properties, equilibria as a solution concept, restricting to finite-state strategies) and some are also used to make the proofs practicable (i.e., using strict equilibria rather than ordinary equilibria, see the discussion in the conclusion). 

\begin{theorem}
	\label{thm:NE}
	The following problem is in 2\exptime: given a $\LEX(\LTL,\mp)$ game $G$ and a rational $\epsilon \geq 0$, decide whether there exists a strategy profile $\vec{\upsigma}$ such that $\vec{\upsigma} \in \SEPNE(G)$.
\end{theorem}

We split the proof into four steps, which we now outline. After giving the technicalities of these steps, we show how to put them together in Section~\ref{sec:together} to establish the theorem.

\paragraph{Replace LTL- by parity-objectives.} 
	In Section~\ref{sec:LTLtoParity} we show that every $\LEX(\LTL, \mp)$ game can be converted into a $\LEX(\parity, \mp)$ game having the same set of finite-state strict $\epsilon$-Nash Equilibria.  
	Intuitively, we replace each agent's \LTL objective by a parity objective.
	Let $P$ be a finite set of integers - then a sequence of priorities, $p_0 p_1 \cdots \in P^\omega$, satisfies the parity condition if the largest priority occurring infinitely often in the sequence is even.
	Thus, the rest of the proof applies to a $\LEX(\parity,\mp)$-game $G = (A,(\kappa_a)_{a \in \Ag}, (\rho_a)_{a \in Ag})$ where $\rho_a: \St \to \int$ is a priority function for each agent, and agent $a$'s primary objective is to ensure the sequence $\rho_a(\cdot)$ satisfies the parity condition, and $a$'s secondary objective is to maximise the mean-payoff of $\kappa_a(\cdot)$.
	Later, we show that such a reduction results in a $\LEX(\parity, \mp)$ game that is at most doubly exponentially larger than the size of the goals $\gamma_{a}$'s.

\paragraph{Two-agent zero-sum games.}	
In Section~\ref{sec:approxvalues} we study two-agent zero-sum games with a $\LEX(\parity, \mp)$ objective (played on the same arena as $G$). We prove that every such game, $H$, has a minimax value $val(H)$, and this value is computable in time polynomial in the number of states and edges, and exponential in the number of priorities and weights of the game. Moreover, we show that for every $\epsilon > 0$ there exists a finite-state strategy for the minimizing agent that ensures the maximizing agent's payoff can achieve a payoff of at most $val(H) + \epsilon$. The proofs in this section make use of Mean-Payoff Parity Games. In particular, it builds from the computation of optimal value of these games as in~\cite{CHJ05} to derive the optimal value in the $\LEX(\parity,\mp)$ setting.

\paragraph{Reducing Equilibrium Finding to Path Finding}	
In Section~\ref{sec:characterisingSNE} we reduce the problem of $\SEPNE$-emptiness to the one of finding payoff thresholds $\tup{z} \in \Omega^{\Ag}$ and an ultimately periodic path $\pi$ in a certain graph (that we call $G[\tup{z}]$) such that $z_a \lex \pay_a(\pi) + \epsilon$. More precisely, each $z_a$ is a so-called ``punishing value'', {\em i.e.}, the value of a two-agent zero-sum game with a $\LEX(\LTL,\mp)$ objective played on the same arena as $G$, starting at some state $s \in St$, but with $a$ trying to maximise its payoff and the rest of the opponents (viewed as a single player) trying to minimise $a$'s payoff. As such, we also show that for each agent $a$, the value $z_a$ can be taken from the set of values (parameterised over the possible start states of the arena) of the two-player game $H$ considered in the previous step. Thus, the state space of the vector $\tup{z}$ is bounded by the number of states of the game. 
	
\paragraph{Path Finding in Multi-Weighted Graphs with LEX(parity,mp) Payoffs}
In Section~\ref{sec:pathfinding} 
	we show how to find ultimately periodic paths $\pi$ such that $z_a  \lex \pay_a(\pi) + \epsilon$ in graphs of the form $G[\bar{z}]$.
	We do this by adapting the linear programming approach for computing zero-cycles in mean-payoff graphs~\cite{KS88}.

\subsection{Replacing \LTL objectives by parity objectives} \label{sec:LTLtoParity}

In this section, we show that every $\LEX(\LTL, \mp)$ game can be converted into a $\LEX(\parity, \mp)$ game having the same set of finite-state strict $\epsilon$-Nash Equilibria.  We begin with a definition of the parity condition and $\LEX(\parity,\mp)$-games.

\paragraph{Parity games}
A sequence $\alpha \in X^\omega$, where $X$ is a finite non-empty set of integer \emph{priorities}
satisfies the \emph{parity condition} if the largest priority occurring infinitely often is even. 

A \emph{$\LEX(\parity,\mp)$ game} $G$ is a tuple 	\[\tpl{A, (\upkappa_{a})_{a \in \Ag}, (\rho_{a})_{a \in \Ag} }\] where $A$ is an arena, $\kappa_a:\St \to \int$  is a weight function for agent $a$, 
and $\rho_a:\St \to \int$ is a \emph{priority function} for agent $a$. For an execution $\pi = s_0 \dec_0 s_1 \dec_1 \cdots \in \exec$ let $\rho(\pi) = \rho(s_0) \rho(s_1) \cdots \in \int^\omega$. For an agent $a \in \Ag$ define 
\[
\parity_a(\pi) = \begin{cases} 
\top & \text{ if } \rho(\pi)\text{ satisfies the parity condition.}\\
\bot & \text{ otherwise.}
\end{cases}  
\]
The payoff function $\pay_a: \exec \to \Omega$ for agent $a$ is defined by 
\[\pay_a(\pi) = (\parity_a(\pi),\mp_a(\pi)).\]
As before, each agent is trying to maximise its payoff under the lexicographic ordering. This completes the definition of parity games.

\paragraph{Deterministic Automata}

In order to systematically perform the required conversion, we introduce \emph{deterministic parity automata}. An \emph{automaton} is a tuple $\tpl{\Sigma,Q,q_0,\Delta}$ where $\Sigma$ is a finite non-empty set called the \emph{input alphabet}, $Q$ is a finite non-empty set of \emph{states}, $q_0 \in Q$ is an \emph{initial state}, and $\Delta:Q \times \Sigma \to Q$ is a \emph{transition function}. 

A \emph{deterministic parity automaton on words (DPW)} is an automaton with a \emph{priority function} $\rho:Q \to \mathbb{Z}$. The \emph{number of priorities} is the cardinality of the set $\rho(Q)$. An \emph{input word} is an infinite sequence over $\Sigma$. A \emph{run} is an infinite sequence over $Q$. Every input word $w_0 w_1 \cdots$ determines a run $s_0 s_1 \cdots$, i.e., $s_0 = q_0$ and $\Delta(s_i,w_i) = s_{i+1}$ for every $i \geq 0$. A run is \emph{accepting} if the largest priority occurring infinitely often in $\rho(s_0) \rho(s_1) \cdots$ is even. An input word is \emph{accepted} if its run is accepting. The \emph{language} of a DPW $\Automaton$, denoted $\Language(\Automaton)$, is the set of input words it accepts. 

With this machinery in place, one can effectively compile \LTL formulas into DPW:
\begin{theorem}~\cite{Var95,Pit07} \label{thm:LTL to DPW}
One can effectively transform a given $\LTL$ formula $\varphi$ over $\AP$ into a DPW 
$\Automaton = \tpl{2^{\AP}, Q, q^{0}, \Delta, \rho}$ over the alphabet $2^{\AP}$ such that $\Language(\Automaton) = \{\alpha \in (2^{\AP})^\omega: \alpha \models \varphi \}$. Moreover, the number of states of the DPW is at most doubly exponentially larger than the size of $\varphi$, and the number of priorities is at most singly exponentially larger than the size of $\varphi$.
\end{theorem}

Thus, we start by translating every \LTL goal $\gamma_a$ into a deterministic parity word (DPW) automaton $\Automaton_a = \tpl{2^\AP, Q_a, q_a^0, \Delta_a, \rho_a}$. Then, for a given $\LEX(\LTL, \mp)$ game $G = \tpl{A, (\upkappa_{a})_{a \in \Ag}, \AP, \lab, (\upgamma_{a})_{a \in \Ag}}$ over the arena $A = \tpl{\Ag,  \Act, \St, \iota, \tr}$, define the arena $A' = \tpl{\Ag,  \Act, \St', \iota', \tr'}$ where
\begin{itemize}
	\item
	$\St' = \St \times \prod_{a \in \Ag} Q_{a}$;
	
	\item
	$\iota' = (\iota, q_{a_1}^0, \ldots, q_{a_n}^0)$;
	
	\item
	For each state $(s, q_{a_{1}}, \ldots, q_{a_{n}})$ in $\St'$ and decision $\delta \in \Act^{\Ag}$, let 
	\[ \tr'((s, q_{a_{1}}, \ldots, q_{a_{n}}), \delta) = (\tr(s, \delta), \Delta_{a_{1}}(q_{a_{1}}, \lab(s)), \ldots, \Delta_{a_{n}}(q_{a_{n}}, \lab(s))).\]
\end{itemize}

Define the $\LEX(\parity, \mp)$ game $G' = \tpl{A', (\upkappa'_{a})_{a \in \Ag}, (\rho'_{a})_{a \in \Ag}}$ where 
\begin{align*}
\upkappa'_{a}(s, q_{a_{1}}, \ldots, q_{a_{n}}) & = \upkappa_{a}(s),\\
\rho'_{a}(s, q_{a_{1}}, \ldots, q_{a_{n}}) & = \rho_{a}(q_{a}).
\end{align*}

Intuitively, the $\LEX(\parity, \mp)$ game $G'$ is the product of the $\LEX(\LTL, \mp)$ game $G$ and the collection of parity automata that recognize the models of each player's $\LTL$ goal. 
Informally, the game executes the original game in parallel with the automata at every step of the game, the arena-component of the product state follows the transition function of the original game~$G$, while the automata-components follow the transition functions of the simulated automata and are updated according to the labelling of the current state of $G$.
As a result, the execution in $G'$ can be recovered from the original execution $\pi$ in the game~$G$ and the unique runs of the (deterministic) automata generated when reading the word $\lab(\pi)$.

Observe that in the translation from~$G$ to its associated~$G'$ the set of actions for each player is unchanged.
This, in turn, means that the set of strategies in both $G$ and $G'$ is the same; indeed, recall from the definitions that strategies are functions from finite sequences of decisions (not states) to actions.
Using this correspondence between strategies in~$G$ and strategies in~$G'$, we can prove the following Proposition, which states an invariance result between $G$ and $G'$ with respect to the satisfaction of players' goals. 

\begin{proposition}[Payoff invariance]
	\label{prp:payinvariance}
	Let $G$ be a $\LEX(\LTL, \mp)$ game and $G'$ its associated $\LEX(\parity, \mp)$ game.
	Then, for every strategy profile $\vec{\upsigma}$ and player $a$, it is the case that $\pay_a(\pi^{G}_{\vec{\upsigma}}) =  \pay_a(\pi^{G'}_{\vec{\upsigma}})$, where by $\pay_a(\pi^{G}_{\vec{\upsigma}})$ we denote the payoff of agent $a$ on the execution $\pi^{G}_{\vec{\upsigma}}$ in $G$ and $\pay_a(\pi^{G'}_{\vec{\upsigma}})$ the payoff of agent $a$ on the execution $\pi^{G'}_{\vec{\upsigma}}$ in $G'$.
\end{proposition}

\begin{proof}
	We will use the following notation: $\pi = \pi^G_{\vec{\upsigma}}$ and $\pi' = \pi^{G'}_{\vec{\upsigma}}$. It is sufficient to show, for every agent $a$, that $\kappa_a(\pi) = \kappa'_a(\pi')$ and that $\pi \models \gamma_{a}$ iff $\rho'_a(\pi')$ satisfies the parity condition. We use the following fact, which follows by the construction of the game $G'$: there is an exact two-way correspondence between runs in $G$, and runs in $G^\prime$. If $\pi = s_0 \dec_0 s_1 \dec_1 \ldots$, then $\pi' = (s_0,\tup{q}^0) \dec_0 (s_1,\tup{q}^1) \dec_1 \ldots$, where $\tup{q}^i = (q_{a_1}^i,\ldots, q_{a_n}^i) \in \prod_a Q_a$ has the property that $r_a = q_a^0 q_a^1 \ldots$ is the unique run of the DPW $\Automaton_a$ for \LTL goal $\gamma_a$ on input $\lambda(s_0) \lambda(s_1) \cdots$ (for every agent $a$). Conversely, if $\pi^\prime$ is a run in $G^\prime$, then it is easy to see that $\pi$ is the unique run in $G$ where the induced sequence of sets of atomic propositions that arises corresponds to the running of the deterministic parity automata.
	
	Now, for the quantitative part of the payoff, note that by definition of $\kappa'$, $\kappa_a(\pi) = \kappa'_a(\pi')$, as required. Similarly, for the qualitative part of the payoff, note that $\pi \models \gamma_a$ iff the run $r_a$ is accepting (since this is how $\Automaton_a$ was chosen), i.e., $\rho_a(r_a)$ satisfies the parity condition. But $\rho_a(r_a) = \rho'_a(\pi')$ by definition of $\rho'$.
	\hfill \qed
\end{proof}

Proposition~\ref{prp:payinvariance} allows us to prove that the set of strict $\epsilon$-Nash Equilibria in $G$ exactly corresponds to the set of strict $\epsilon$-Nash Equilibria in $G'$.
The following result holds.

\begin{proposition} \label{prop:LTLtoParity}
	Let $G$ be a $\LEX(\LTL, \mp)$ game and $G'$ its associated $\LEX(\parity, \mp)$ game.
	Then $\SEPNE(G) = \SEPNE(G')$.
\end{proposition}

\begin{proof}
	We show one direction (the other is symmetric). Let $\vec{\upsigma}$ be a strategy profile that is not in $\SEPNE(G')$, i.e., there is an agent $a$ and a strategy profile $\vec{\upsigma}'$ such that $\vec{\upsigma}'_a \neq \sigma_a$, $\vec{\upsigma}'_b = \sigma_b$ for $b \neq a$, and $\pay_{a}(\pi^{G'}_{\vec{\upsigma}}) + \epsilon \lexeq \pay_{a}(\pi^{G'}_{\vec{\upsigma}'})$.
	Thus, by applying Proposition~\ref{prp:payinvariance} on both sides of the inequality, we obtain that $\pay_{a}(\pi^G_{\vec{\upsigma}}) + \epsilon \lexeq \pay_{a}(\pi^G_{\vec{\upsigma}'})$, which implies that 
	$\vec{\upsigma} \not \in \SEPNE(G)$.  \qed
\end{proof}

Finally, note that the number of states of $G^\prime$ is doubly exponential in the size of the LTL goals of $G$, and that the number of priorities in $G^\prime$ is singly exponential in the size of the LTL goals of $G$. This will be important when establishing the complexity of the existence problem later in the paper.

\subsection{Two-Agent Zero-Sum Lex(parity,mp)-Games} \label{sec:approxvalues}

We begin with a study of two-agent zero-sum $\LEX(\parity, \mp)$ games.
To simplify notation, we define these as $H = \tpl{A, \upkappa, \uprho}$ where $A$ is an arena with $\Ag = \{1, 2\}$, and $\upkappa,\uprho:\St \to \int$. 
Define $\pay(\pi) = (\parity(\pi), \mp(\pi))$.
Player $1$ is called the ``maximizer'' and player $2$ is called the ``minimizer''. Thus, intuitively, agent $1$ is trying to maximise the value of $\pay(\cdot)$ while agent $2$ is trying to minimize it.

Now, a basic question is; `what is the highest payoff that player 1 can achieve?'. Formally, we can ask what the following quantity is:
\begin{equation*}
	\overline{val(H)} = \sup_{\upsigma}\inf_{\zeta} \pay(\pi_{\langle \upsigma, \zeta \rangle}).
\end{equation*}

Here, $\upsigma$ ranges over strategies of player 1 (the maximizer), $\zeta$ ranges over strategies of player 2 (the minimizer), and $\pi_{\langle\upsigma,\zeta \rangle}$ is the unique execution determined by the strategy profile $\langle\upsigma,\zeta\rangle$. 

Conversely, we can consider the smallest payoff that player 2 can inflict on player 1. That is, the quantity
\begin{equation*}
	\underline{val(H)} = \inf_{\zeta}\sup_{\upsigma} \pay(\pi_{\langle \upsigma, \zeta \rangle}),
\end{equation*}

where $\upsigma, \zeta$ and $\pi_{\langle \upsigma, \zeta\rangle}$ are as above.

In arbitrary, zero-sum two player games, we'd call these two values the maximin and minimax values, and as in arbitrary games, it is easy to verify that $\underline{val(H)} \leq \overline{val(H)}$. Thus, a natural question for such games is under what conditions do these two values coincide -- when do we have $\underline{val(H)} = \overline{val(H)}$? If they do coincide, then we call this the \emph{value} of the game and denote it by $val(H)$. We show that for zero-sum Lex(parity, mp)-games, these two values do indeed coincide, and so the value is well-defined. Moreover, we show it can be computed in \tfnp, the class of total function problems which are solvable in nondeterministic polynomial time \cite{DBLP:journals/tcs/MegiddoP91}. 

To define \tfnp formally, suppose we have an alphabet $\Sigma$ and a left-total binary relation $R: \Sigma^* \times \Sigma^*$ such that for all $x,y \in \Sigma^*$, $R(x,y)$ implies that $\abs{y} \leq p(\abs{x})$ for some polynomial $p$. Moreover, suppose that given $x, y \in \Sigma^*$, we can determine in polynomial time whether $R(x,y)$ holds. Then a natural problem is to ask `given an $x \in \Sigma^*$, find a $y \in \Sigma^*$ such that $R(x,y)$'. The class of all such problems is exactly \tfnp.   

With this terminology in place, we are in a position to prove the following proposition.

\begin{proposition}\label{prop:value}
	Every two-agent zero-sum $\LEX(\parity,\mp)$ game $H$ has a value, denoted 
	$val(H) \in \Omega$. Moreover, this value can be computed in 
	\tfnp.
\end{proposition}

\begin{proof}
	W.l.o.g., we can consider $H$ to be turn-based. Indeed, we can ensure this by replacing every transition $s \xrightarrow{(c_1,c_2)} s'$ by two transitions $s \xrightarrow{c_1} s_{c_{1}} \xrightarrow{c_2} s'$, in a way that all the original states belong to Player 1, while every extra state $s_{c}$ belongs to Player 2, and has the same weight and priority as $s$. Note that such a construction depends on the ordering of players, {\em i.e.}, in order to compute the value for Player 2, we need to employ a construction of a game $H''$ that replaces $s \xrightarrow{(c_1,c_2)} s'$ by $s \xrightarrow{c_2} s_{c_{2}} \xrightarrow{c_1} s'$. It is easy to verify that if we have a run $\pi$ in a non-turn-based game, $H$, and transform it to a turn-based game, $H^\prime$, then the corresponding run $\pi^\prime$ will induce the same payoff as in $\pi$. 
		
	We compute $\overline{val(H)}$ and $\underline{val(H)}$ by reducing to solving two-agent turn-based zero-sum games $K$ with mean-payoff parity objectives~\cite{CHJ05}, which are known to be in \np~\cite{BCHJ09}. We show that these two values are equal, and thus, $val(H)$ is well defined. Moreover, we require two parallel invocations of an \np algorithm, and so the whole process lies within \tfnp.
	
	The mean-payoff parity games $K$ are played on the same weighted arenas $\left<A,\upkappa,\uprho\right>$ as two-agent zero sum $\LEX(\parity,\mp)$ games. However, the payoff set for $K$ is $\overline{\mathbb{R}} = \mathbb{R} \cup \{\pm \infty\}$ with its usual ordering $<$, and payoff function $\pay^+:\exec \to \overline{\mathbb{R}}$ is defined as follows: $\pay^+(\pi)$ equals $-\infty$ if $\parity(\pi) = \bot$, and $\mp(\pi)$ otherwise. 
	Informally, the first player is trying to satisfy the parity condition, and once that holds, maximise its mean-payoff. The \emph{value} $val(K)$ is defined to be the maximum payoff that the first player can achieve. It follows from Theorems 2 and 3 of ~\cite{CHJ05} that values exist for these games and can be computed. Moreover, such computation requires nondeterministic polynomial time and lies in \np \cite{BCHJ09}.
	
	To help us calculate the two values of interest, we also introduce an auxillary game, which we call the \emph{dual} of $K$, which we denote $K^*$. In $K^*$, the first player tries to satisfy the parity condition and if this \emph{doesn't} hold, then tries to maximise their mean-payoff. Formally, these are games $K^* = \left<A,\upkappa,\uprho\right>$ with payoff function defined as follows: $\pay^*(\pi)$ equals $\infty$ if $\parity(\pi) = \top$, and $\mp(\pi)$ otherwise. 
	
	We argue that these auxillary games also have a well-defined value, $val(K^*)$. From player 2's perspective, their payoff is $-
	\infty$ if the parity condition holds and $-\mp(\pi)$ otherwise. Suppose we negated all the weights, and added 1 to each priority. Then in this new game, player 2 is trying to first satisfy the parity condition, then trying to maximise $\limsup_{n \to \infty}(\avg_n(\kappa(\pi)))$. Thus, we can \emph{almost} view $K^*$ as a mean-payoff parity game, but not quite. Indeed, we certainly can't directly reduce it to a mean-payoff parity game as it stands. However, consider the line of argument in \cite{CHJ05} that leads to Theorems 2 and 3 -- for every line in this proof, if we use the $\limsup$ rather than the $\liminf$ for player $1$, the results still hold. Thus, we can conclude that these auxiliary games also have a value.
	
	With this machinery in place, we compute the value of two-agent zero-sum $\LEX(\parity, \mp)$ games as follows.
	Let $K$ be the mean-payoff parity game on the same weighted arena as $H$. Then we are in exactly one of two scenarios: either $val(K) \neq -\infty$, or $val(K) = -\infty$.
	
	First suppose that $val(K) \neq -\infty$. Thus, player 1, using some strategy $\upsigma$, can ensure the parity condition is satisfied, and given this, the greatest payoff they can attain is $val(K)$. This is turn implies that $\overline{val(H)} = (\top, val(K))$. Additionally, $val(K) \neq -\infty$ also implies that player 2 cannot force the parity condition to not hold, and the lowest payoff they can inflict on player 1 is $val(K)$. This implies that $\underline{val(H)} = (\top, val(K))$. Putting this together, we get $\overline{val(H)} = \underline{val(H)}$. Thus, $val(H)$ exists and is equal to $(\top, val(K))$. 

	Now suppose that $val(K) = -\infty$. If this is the case, then we necessarily have $val(K^*) \neq -\infty$. Thus, player 2 can force that the largest priority occurring infinitely often is odd and can also ensure the smallest payoff player 1 achieves is $val(K^*)$. Thus, we have $\underline{val(H)} = (\bot, val(K^*))$. Similarly, player 1 cannot force the parity condition to be true, and given this, the greatest payoff they can achieve if $val(K^*)$. Thus, we have $\overline{val(H)} = (\bot, val(K^*))$. Again, this implies that $val(H)$ exists and is equal to $(\bot, val(K^*))$

	Regarding the complexity, observe that the construction of the games $K$ and $K^*$ is linear in the size of $H$, and that we employ an \tfnp procedure to solve them. Once we have calculated the values of both, this will tell us the value of $H$. Since the two procedures are independent of one another, we can do both of them sequentially and then compare their values afterwards. Thus, this guarantees that the overall complexity computing $val(H)$ is \tfnp. \qed
\end{proof}

It is worth noting that in mean-payoff parity games, computing the value of the game can be done in time $O\left(n^d \cdot \left(m + \textsf{Parity} + \textsf{MP}\right)\right)$ \cite{CHJ05}. Here, $n$ is the number of vertices in the game graph, $m$ the number of edges and $d$ the number of priorities. Aditionally, \textsf{Parity} and \textsf{MP} are the complexities of solving parity and mean-payoff games respectively. We will use this fact later when we calculate the overall complexity of the \SEPNE-existence problem.

It is not hard to see that in two-player zero-sum $\LEX(\parity,\mp)$-games $H$, a player may need infinite memory to achieve the optimal value $val(H)$ (Cfr.~\cite[Figure $1$]{CHJ05}). However, as proven in~\cite{BCHJ09}, for every mean-payoff parity game $K$ and every $\epsilon > 0$, there exists a finite-state strategy for the minimizer, $\zeta$ (that depends on $\epsilon$) such that for every strategy $\upsigma$ of the maximizer, it holds that $\pay(\pi_{\upsigma,\zeta}) \leq val(K) + \epsilon$. Thus, using the same argument as in Proposition~\ref{prop:value}, we get:

\begin{proposition}\label{prop:approx}
	For every two-agent zero-sum $\LEX(\parity,\mp)$ game $H$ and every $\epsilon > 0$ there exists a finite-state strategy $\zeta$ for the minimizer, such that for every strategy $\upsigma$ of the maximizer (not necessarily finite-state), it holds that $\pay(\pi_{\upsigma,\zeta}) \lexeq val(H) + \epsilon $.
\end{proposition}

\subsection{Reducing Equilibrium Finding to Path Finding} \label{sec:characterisingSNE}

In this section we show that a path $\pi$ is generated by some $\vec{\upsigma} \in \SEPNE(G)$ iff $\pi$ exists in a certain subgraph of the weighted arena of $G$.
To do this, we adapt the proof in~\cite[Section $6$]{UW11} that shows how to decide the existence of a (not necessarily finite-state) Nash equilibrium 
for mean-payoff games.

We first need the notion of punishing values and strategies.
For $a \in \Ag$ and $s \in \St$ define the \emph{punishing value} $p_a(s)$ to be the $\lexeq$-largest $(x,y)$ that player $a$ can achieve from state $s$ by playing against the coalition $\Ag \setminus \{a\}$, i.e., by turning the game into a two-player zero-sum game in which the maximizer simulates the moves of player $a$, the minimizer simulates the moves of the coalition $\Ag \setminus \{a\}$, and the payoff is that of player $a$. These values can be computed for every player in each state by constructing the appropriate two-agent Lex(parity, mp)-game, $H$, and invoking Proposition~\ref{prop:value}. Formally, if $G = \langle A, (\kappa_a), (w_a) \rangle$, and we want to calculate the punishing value for player $i$ in state $s^\prime$, then $H = \langle A^\prime, \kappa_i, w_i \rangle$, where $\Ag^\prime = \{ i, \Ag \setminus \{i\} \}$, $\Act^\prime = \Act$, $\St^\prime = \St$, $\iota^\prime = s^\prime$ and $\tr^\prime(s, (\dec_i, (\dec_1, \ldots, \dec_{i-1}, \dec_{i+1}, \ldots, \dec_\Ag))) = \tr(s, (\dec_1, \ldots, \dec_\Ag))$. 

Moreover, for every state $s \in \St$, every player $a \in \Ag$ and every $\epsilon > 0$, fix $\zeta_{s,a}^{\epsilon}$ to be the strategy of the minimizer described by Proposition~\ref{prop:approx}. We view $\zeta_{s,a}^{\epsilon}$ as a profile, {\em i.e.}, $\zeta_{s,a}^{\epsilon}:\Ag \setminus \{a\} \to (\hist \to \Act)$, and call $\zeta_{s,a}^{\epsilon}(b)$ an \emph{$\epsilon$-punishing strategy} for agent $b$. Note that these $\epsilon$-punishing strategies are finite-state.

\begin{definition}[Secure values]\footnote{This definition extends the one provided in~\cite{UW11} which considers mean-payoff games.}
	\label{def:zsecure}
	For an agent $a \in \Ag$ and $z \in \Omega$, a pair $(s,\dec) \in \St \times \Act^{\Ag}$ is \emph{$z$-secure for $a$} if $p_a(\tr(s,\dec')) \lexeq z$ for every $\dec' \in \Act^{\Ag}$ that agrees with $\dec$ except possibly at $a$.
\end{definition}

With this definition in place, we can now state the following result. 

\begin{proposition}
	\label{prop:NEchar}
	For every \LEX(\parity, \mp) game $G$, constant $\epsilon \geq 0$, and ultimately periodic path $\pi = s_0 \dec_0 s_1 \dec_1 \ldots$ in $G$, the following are equivalent:
	\begin{enumerate}
		\item There exists $\vec{\upsigma} \in \SEPNE(G)$ such that $\pi = \pi_{\vec{\upsigma}}$.
		\item There exists $\bar{z} \in \Omega^{|\Ag|}$, where 
		$z_a \in \{p_a(s) : s \in \St\}$, $a \in \Ag$, such that for every agent $a$, 
		\begin{enumerate}
			\item
			$z_a \lex \pay_a(\pi) + \epsilon$ and 
			\item
			for all $i \in \mathbb{N}$, the pair $(s_i,\dec_i)$ is $z_a$-secure for $a$.
		\end{enumerate} 
	\end{enumerate}
\end{proposition}

\begin{proof}
	Fix a game $G$, constant $\epsilon \geq 0$ and ultimately periodic path $\pi = s_0 \dec_0 s_1 \dec_1 \ldots$.
	
	For (1) implies (2),  suppose there exists $\vec{\upsigma} \in \SEPNE(G)$ with $\pi = \pi_{\vec{\upsigma}}$.
	Define $\bar{z} \in \Omega^{|\Ag|}$ by $z_a = \max\{p_a(\tr(s_n, \dec'_n)) : n \in \mathbb{N}, \wedge_{b \neq a} \dec'_n(b) = \dec_n(b)\}$,
	{\em i.e.}, $z_a$ is the largest value player $a$ can get by deviating from $\pi$.
	For every $n \in \mathbb{N}$, $(s_n, \dec_n)$ is $z_a$-secure for $a$ (by definition of $z_a$ and $z_a$-secure).
	Moreover, $z_a \lex \pay_a(\pi) + \epsilon$: indeed, let $n$ be such that $z_a = p_a(\tr(s_n, \dec'_n))$, and suppose that $\pay_a(\pi) + \epsilon \lexeq z_a$; then player $a$ would deviate at step $n$ by playing $\dec'_n(a)$ and following a strategy that achieves at least $z_a$ from this point.
	Note that such a (possibly infinite-state) strategy exists by Proposition~\ref{prop:value}.
	But, due to prefix-independence of the payoff function, this is also the payoff of the whole play, contradicting the choice of $\pi$ as the execution of a strict $\epsilon$ Nash-equilibrium. 
	
	For (2) implies (1), let $\bar{z} \in \Omega^{|\Ag|}$ be given with the stated properties. 
	We build a strict $\epsilon$ Nash-equilibrium $\vec{\upsigma}$ such that $\pi_{\vec{\upsigma}} = \pi$. 		
	For $b \in \Ag$, we define $\vec{\upsigma}(b)$ as follows.
 	For every history $h = \dec_0 \ldots \dec_n$ ({\em i.e.}, a decision prefix of $\pi$), define $\upsigma_{b}(h) = \dec_{n}(b)$. Thus, $\vec{\upsigma}$ follows $\pi$ as long as no-one has deviated from $\pi$.

	For every other history, $h^\prime = \dec_0^\prime \ldots \dec_n^\prime$, let $k$ be the first such integer with $\dec_k^\prime \neq \dec_k$. There are two cases - either $\dec_k$ differs in one position, or multiple positions. If $\dec_k$ differs in one position, say by player $a$, then let $s_{k+1}^\prime = \tr(s_k, \dec_k^\prime)$ and then for all other players $b$, set $\sigma_b(h) = \zeta_{s_{k+1}^\prime, a}^\epsilon(b)(h)$. If $\dec_k$ differs in multiple positions, then set $\sigma_b(h)$ arbitrarily. 

	By construction, we have that $\pi_{\vec{\upsigma}} = \pi$. We now aim to show that $\vec{\upsigma} \in \SEPNE(G)$. Suppose it is not. Then some player $a$ has a strategy, $\sigma_i$, such that $\pay_a(\pi_{\vec{\upsigma}}) + \epsilon \lexeq \pay_a(\pi_{(\upsigma_{-a}, \sigma_i)})$. By assumption, this implies that $z_a \lex \pay_a(\pi_{(\upsigma_{-a}, \sigma_i)})$. Moreover, we have that $\pi_{(\upsigma_{-a}, \sigma_i)} \neq \pi_{\vec{\upsigma}}$. Thus, let $(s_j, \dec_j^\prime)$ be the first pair from the execution of $\pi_{(\upsigma_{-a}, \sigma_i)}$ that differs from the execution of $\pi_{\vec{\upsigma}}$ (note that the $j^{\text{th}}$ state of both executions is the same). Additionally, let $\tr(s_j, \dec_j^\prime) = s_{j+1}^\prime$.
	
	Now, by assumption, the pair $(s_j, \dec_j)$ is $z_a$-secure for $a$. This implies that $p_a(s_{j+1}^\prime) \lexeq z_a$, in turn implying we have $p_a(s_{j+1}^\prime) \lex \pay_a(\pi_{(\upsigma_{-a}, \sigma_i)})$. However, by construction, for all players $b \neq a$, we have $\sigma_b(h) = \sigma_{s_{j+1}^\prime, a}^\epsilon(b)(h)$ for all histories with the prefix $\dec_0\ldots\dec_{j-1}\dec_j^\prime$. By Proposition~\ref{prop:approx}, this implies that $\pay_a(\pi_{(\upsigma_{-a}, \sigma_i)}) \lex p_a(s_{j+1}^\prime)$. But then we can conclude that $p_a(s_{j+1}^\prime) \lex p_a(s_{j+1}^\prime)$, which is a contradiction. Thus, we may conclude that $\vec{\upsigma} \in \SEPNE(G)$.
	
	Finally, it is easy to verify that $\vec{\upsigma}$ is a finite-state strategy. Since $\pi$ is ultimately periodic, the `main body' of $\vec{\upsigma}$ is finite-state. Moreover, tracking whether any player has deviated yet requires only finite memory. If a player has deviated, then a punishing strategy is used by all remaining players, which is also finite-state.\qed
\end{proof}

\subsection{Path Finding in Multi-Weighted Graphs with LEX(parity,mp) Payoffs} \label{sec:pathfinding} 

The following theorem, of interest in its own right, will be used to decide the existence of ultimately periodic paths in Proposition~\ref{prop:NEchar}.
A \emph{multi-weighted graph} is a structure of the form $\mathcal{G} = (V,E,(w_a)_{a \in A}, (p_a)_{a \in A})$ where $V$ is a finite set of states, $E \subseteq V^2$ a set of edges, $A$ is a finite index set, and $w_a,p_a: V \to \int$ are functions, one for each $a \in A$, mapping states to integers.

\begin{theorem}
	\label{thm:pathfind}
	Given a multi-weighted graph $\mathcal{G} = (V, E, A, (w_a)_{a \in A}, (p_a)_{a \in A})$ over the finite index set $A$, a starting vertex $\init \in V$, and a vector of payoffs $f \in \Omega^{A}$, one can decide in nondeterministic polynomial time whether there exists an ultimately periodic path $\pi = v_0 v_1 \cdots$ in the graph with $\init = v_0$ and, for every $a \in A$, $f_a \lex \pay_a(\pi)$.
\end{theorem}

\begin{proof}	
	W.l.o.g., we may assume that $f_a \in \{\top,\bot\} \times \{0\}$ (to see this, redefine $w_a(s)$ to be $w_a(s) - f_a$ for all $s \in V, a \in A$). Also, we may assume that every state in $V$ is reachable from $\init$ (to see this, restrict $V$ to the states reachable from $\init$, computable in linear time). Finally, nondeterministically guess a vector $P \in \int^A$. This vector represents the top priorities visited infinitely often for each index $a \in A$, and thus also determines the exact form of $f$. 

	Consider the subgraph, $\mathcal{G}^\prime$ formed by iterating through each $a \in A$ and all the states, and removing those states with a priority higher than $P_a$ for some $a$. Now, if the original graph $\mathcal{G}$ has some ultimately periodic path $\pi$ with the top priorities being visited infinitely often given by $P$, then $\pi$ is also a path in $\mathcal{G}^\prime$. Thus, it suffices to form $\mathcal{G}^\prime$ and ask if there is some path $\pi$ such that $f_a \lex \pay_a(\pi)$. In what follows, we simply relabel $\mathcal{G}^\prime$ as $\mathcal{G}$, on the understanding that the above transformation has taken place.
	
	We now reduce the problem to finding certain cycles in $\mathcal{G}$. A \emph{cycle} is 
	a finite path $C$ of the form $s_0 s_1 \cdots s_n$ (for some $n \geq 1$) such 
	that $s_0 = s_n$ (note that cycles are not necessarily simple). 
	Write $s \in C$ to mean that $s = s_i$ for some $i \leq n$. 
	Define $\tot_a(C) = \Sigma_{j = 1}^n w_a(s_j)$, $\avg_a(C) = \frac{\tot_a(C)}{n}$, and $\max_a(C) = \max_{1 \leq j \leq n} p_a(s_j)$.  
	The stated problem is equivalent to deciding, given $W$ and $h:D \to \int$, 
	if there exists a cycle $C$ in $W$ such that i) $\max_a(C) = h(a)$ for every 
	$a \in D$, ii) $\avg_a(C) > 0$ for every $a \in A$. Note that we can replace 
	$\avg_a(C)$ by $\tot_a(C)$ in this  problem (since $\avg_a(C) > 0$ iff 
	$\tot_a(C) > 0$).
	In order to decide the existence of such a cycle, we adapt the proof 
	from~\cite{KS88} that shows how to decide if there is a cycle $C$ such that 
	for every $a \in A$, $\tot_a(C) = 0$.
	
	A \emph{multicycle} \MC  is a non-empty multiset of cycles.
	Thus a cycle is a multicycle \MC with $|\MC| = 1$. Extend $\tot_a$ and 
	$\max_a$ to multicycles as follows: $\tot_a(\MC) = \sum_{C \in \MC} 
	\tot_a(C)$ and $\max_a(\MC) = \max \{\max_a(C) : C \in \MC\}$.
	An \emph{$\upeta$-multicycle} is a multicycle \MC such that for all $a \in A$, we have i) $\max_a(\MC) = 
	P_a$, and ii) $\tot_a(\MC) > 0$. Additionally, an \emph{$\upeta$-cycle} is simply an $\upeta$-multicycle, $\MC$, with $\abs{\MC} = 1$ - that is, it consist of a single cycle. 
	
	Thus, deciding the problem started in the theorem is equivalent to deciding if there is an $\upeta$-cycle $\MC$. We now show that it is sufficient to decide if there is an $\upeta$-multicycle.
	
	Define a relation on $V$: $v \equiv w$ iff $v = w$ or there exists an 
	$\upeta$-multicycle $\MC$ and $C \in \MC$ such that $v,w \in C$.
	Note that $\equiv$ is an equivalence relation: indeed, if $u,v \in C$ for $C 
	\in \MC$, and $v,w \in C'$ for $C' \in \MC'$ then $u,w \in C''$ for $C'' \in 
	\MC''$ where $C''$ is formed by tracing $C$ from $v$ to $v$ and then tracing 
	$C'$ from $v$ to $v$, and $\MC''$ is $(\MC \cup \MC' \cup \{C^{\prime\prime}\}) \setminus \{C,C'\}$.
	Note that $\MC''$ is an $\upeta$-multicycle because $\tot_a(C'') = \tot_a(C) 
	+ \tot_a(C')$ and $\max_a(C'') = \max\{\max_a(C),\max_a(C')\}$.
	
	Suppose $\equiv$ has index $1$, {\em i.e.}, for all $v,w \in V$, $v \equiv w$. There are two cases - $\abs{V} = 1$ and $\abs{V} > 1$. First suppose that $\abs{V} = 1$ with $V = \{v\}$. Then either $v$ has a self-loop or it doesn't. If it does not, then there can be no $\upeta$-cycle. If it does, and the weight of the $v$ is strictly positive and its priority even, then it has an $\upeta$-cycle. Otherwise, it does not.
	
	Now suppose that $\abs{V} > 1$. We claim that there exists an $\upeta$-cycle.
	Indeed: for every $v, v' \in V$ let $\MC_{v,v'}$ be an $\upeta$-multicycle containing a cycle $C$ that visits $v$ and $v'$.
	Then $\MC = \cup_{v, v' \in V} \MC_{v,v'}$ is an $\upeta$-multicycle such that (*): for every $v, v' \in V$ there exists $C \in \MC$ such that $v,v' \in C$.
	We now define two transformations of multicycles $\MC \mapsto \MC'$ that maintain the following invariants: a) $\tot_a(\MC) = \tot_a(\MC')$ for $a \in A$, b) $\max_a(\MC) = \max_a(\MC')$ for $a \in D$, c) if $\MC$ satisfies (*) then so does $\MC'$, d) $|\MC'| < |\MC|$ ({\em i.e.}, the number of cycles decreases).
	Thus, repeatedly applying these transformation results in an $\upeta$-cycle. 
	
	First, if $C$ occurs more than once in $\MC$, say $n$ times, then remove all 
	occurrences of $C$ from $\MC$ and add the single cycle formed by tracing $C$ 
	$n$-many times. 
	Thus, we have that $\MC$ is a \emph{set} of cycles ({\em i.e.}, not a proper 
	multiset).
	Second, if \MC is not a single cycle, take $C \neq C' \in \MC$, $v \in C, v' 
	\in C'$ and by (*) pick $D \in \MC$ such that $v,v' \in D$.
	There are three cases: if $D \neq C, D \neq C'$ then form the cycle $F$ by 
	tracing $C$ from $v$ to $v$, then tracing ``half'' of $D$ from $v$ to $v'$, 
	then tracing $C'$ from $v'$ to $v'$, and then tracing the ``other half'' of 
	$D$ from $v'$ to $v$ and let $\MC'$ be $\MC \cup \{F\} \setminus \{C,C',D\}$; 
	if $D = C$ (the case $D = C'$ is symmetric), then $v' \in C$ and thus form 
	$F$ by tracing $C$ from $v'$ to $v'$ and then tracing $C'$ from $v'$ to $v'$, 
	and let $\MC'$ be $(\MC \cup \{F\}) \setminus \{C,C'\}$.
	Both transformations satisfy the invariants.
	
	Thus, the following algorithm decides if there is an $\upeta$-cycle (assuming one can decide if there exists an $\upeta$-multicycle): if $\abs{V} = 1$, check if the single node in $V$ has a self-loop, a strictly positive weight and an even priority.
	If it does, output ``yes'', otherwise, output ``no''.
	If $\abs{V} >1$, compute $\equiv$ for $V$; if it has index $1$ then output ``yes''; else, for each equivalence class $X \in V / \equiv$, recurse on the subgraph induced by $X$.
	The algorithm is clearly sound, {\em i.e.}, if it outputs ``yes'' then there is indeed an $\upeta$-cycle.
	To see that it is complete, note if that $C$ is an $\upeta$-cycle, then for all $v,w \in C$, $v \equiv w$; and thus $C$ is contained in an $\equiv$-class.
	
	Finally, we show how to decide if there exists an $\upeta$-multicycle using linear programming.
	We temporarily make the assumption that all the nodes of the graph lie in the same strongly connected component (SCC). For every edge $e$ introduce a variable $x_e$.
	Informally, the value $x_e$ is the number of times that the edge $e$ is used on an $\upeta$-multicycle.
	Formally, let $\src(e) = \{v \in V : \exists w\, e = (v,w) \in E\}$; $\trg(e) 
	= \{v \in V : \exists w\, e = (w,v) \in E\}$; $\OUT(v) = \{e \in E : \src(e) 
	= v\}$ and $\IN(v) = \{e \in E : \trg(e) = v\}$.
	
	The linear program LP has the following inequalities and equations:
	\begin{enumerate}
		\item[Eq1:]
		$x_e \geq 0$ for each edge $e$ --- this is a basic consistency criterion;
		
		\item[Eq2:]
		$\Sigma_{e \in E} x_e \geq 1$ --- this ensures that at least one edge is chosen;
		
		\item[Eq3:]
		for each $a \in A$, $\Sigma\{ w_a(\src(e)) \cdot x_e : e \in E\} > 0$ --- 
		this enforces that the total sum is positive;
		
		\item[Eq4:]
		for each $a \in D$, $\Sigma \{x_e : p_a(\src(e)) = P_a\} \geq 1$ --- this 
		ensures that the largest appearing priority for agent $a$ is $P_a$;
		
		\item[Eq5:]
		for each $v \in V$, $\Sigma_{e \in \OUT(v)} x_e = \Sigma_{e \in \IN(v)} 
		x_e$  --- this ``preservation'' condition says that the number of times one 
		enters a vertex is equal to the number of times one leaves that vertex.	
	\end{enumerate}

	We now prove that there exists a $\upeta$-multicycle if and only if the LP has an integer solution.
	Indeed, from left to right let $\MC$ a (non-empty) $\upeta$-multicycle and let $x_e$ be the number of occurrences of the edge $e$ in any of the cycles in $\MC$.
	Clearly, Eq1 and Eq2 are satisfied.
	Moreover, since $\MC$ a $\upeta$-multicycle, it holds that both $\sum_a(\MC) > 0$ and $\max_a(\MC) = P_a$ for each $a \in A$, which implies Eq3 and Eq4, respectively.
	Finally, observe that every vertex in a cycle is entered and left an equal number of times. Therefore, being $\MC$ a multicycle implies that Eq5 is satisfied and so the LP has an integer solution.
	
	From right to left, assume the LP has an integer solution $\{x_e\}_{e \in E}$. From Eq1, Eq2 we obtain that the solution provides that each edge $e$ is counted $x_e$ times. Moreover, by Euler's Theorem~\cite[Theorem 1.6.3]{JensenGutin}, Eq5 implies that the graph has an Eulerian cycle. Thus, the edges selected by the solution can be arranged in a multicycle $\MC$. Finally, Eq3 and Eq4 guarantee that $\sum_a(\MC) > 0$ and $\max_a(\MC) = P_a$ for each $a \in A$, thus implying that $\MC$ is a $\upeta$-multicycle.
	
	Now, from~\cite{BG07}, we obtain that the program LP has a solution in the reals iff it has a solution in the rationals~\cite{BG07}.
	Moreover, if $(x_e)_{e \in E}$ is a solution to LP and $k \in \mathbb{N} 
	\setminus \{0\}$, then $(kx_e)_{e \in E}$ is also a solution to LP.
	Thus, the LP has a solution iff it has an integer solution. Thus, the LP has a solution iff the graph has an $\upeta$-multicycle.

	With this in place, we can now relax the assumption of the graph consisting of a strongly connected component. First, we can use Tarjan's algorithm \cite{DBLP:journals/siamcomp/Tarjan72} to obtain the strongly connected components (SCCs) of the graph in linear time. Then for every non-trivial SCC that is reachable from the start node, we can use the above procedure to determine if it contains an $\eta$-cycle $\MC$. As payoffs are prefix-independent with respect to paths, this implies that $\MC$ is an $\eta$-cycle for the whole graph. 

	Now, in terms of complexity, we need to non-deterministically guess a vector $P \in \mathbb{Z}^A$. Then we recursively compute the equivalence relation $\equiv$, which can be done with an application of Tarjan's algorithm, in time $O\left(\abs{V} + \abs{E}\right)$, following by the application of $|E|$ linear programs (Cfr.~\cite{KS88}). Finally, observe that the size of the linear program is polynomial in the size of the graph. By an identical argument to \cite{KS88}, we can conclude that the problem can be determined in non-deterministic polynomial time. \qed
\end{proof}

\subsection{Putting the Steps Together} \label{sec:together}

We can now finish the proof of Theorem~\ref{thm:NE}. Consider a rational $\epsilon \geq 0$ and a $\LEX(\LTL,\mp)$ game $G$. Throughout, let $\abs{\gamma}$ and $W$ denote the size of the largest goal and weight respectively.

First, by Proposition~\ref{prop:LTLtoParity} we can transform $G$ into a $\LEX(\parity, \mp)$-game, $G^\prime$. In this new game, the number of states, $n^\prime$ is at most doubly-exponential in the size of the goals, $\abs{\gamma}$ and the number of priorities, $\abs{\rho}$ is at most singly-exponential in the size of the goals. 

Second, for every agent $a \in \Ag$ and a state $s \in \St$, we compute the punishing value $p_a(s)$ by means of Proposition~\ref{prop:value}. Constructing the corresponding mean-payoff parity games needed to do this can be done in time linear of the size of the game. Moreover, we can determine the value of these games in time $O(n^d \cdot (m + \textsf{MP} + \textsf{Parity}))$, where $n$ is the number of vertices in the game graph, $m$ the number of edges, $d$ is the number of priorities, and $\textsf{MP}$ and $\textsf{Parity}$ denote the complexity of solving mean-payoff and parity games respectively. Solving mean-payoff games can be done in time $O(\abs{V}^3\cdot \abs{E}\cdot W)$, where $V$ is the set of vertices of the game, $E$ is the set of edges, and $W$ is the largest weight \cite{ZP95}. Moreover, parity games can be solved by reducing them to mean-payoff games \cite{Jurd98}. The reduction implies that these games can be solved in time $O(\abs{V}^3\cdot \abs{E}\cdot \abs{V}^{\abs{\rho}})$, where $V$ and $E$ are the vertices and edges of the game graph, and $\abs{\rho}$ is the size of the largest priority. Putting all this together, we see that calculating the punishment values of all players in all states can be done in time,
\begin{align*}
	&\phantom{=}\abs{\Ag}\cdot n^\prime \cdot O\left({n^\prime}^{\abs{\rho}} \cdot \left({n^\prime}^2 + {n^\prime}^3\cdot{n^\prime}^2\cdot W + {n^\prime}^3\cdot{n^\prime}^2\cdot {n^\prime}^{\abs{\rho}} \right)\right) \\
	&= O\left(\abs{\Ag}\cdot W \cdot 2^{2^{q(\abs{\gamma})}}\cdot \right),
\end{align*}

where $q$ is some appropriately chosen polynomial. Thus, the above step can be done in time doubly exponential in the size of the goals of the game and linear in the number of agents and the size of the largest weight. 

Third, thanks to the characterization of finite-state strict $\epsilon$-Nash Equilibria provided in Proposition~\ref{prop:NEchar}, there is a $\vec{\upsigma} \in \SEPNE(G)$ if, and only if, there exists a tuple $\tup{z} \in \mathbb{Z}^{\Ag}$ of values with $z_a \in \{p_a(s) : s \in \St\}$, and a path $\pi$ in $G$ such that $z_a \lex \pay_a(\pi) + \epsilon$ for all $a \in \Ag$, and for all $i \in \SetN$, $\pi_{i} = (s_i, \delta_i)$ is $z_a$-secure for $a$.
Now, let $G[\tup{z}]$ denote the multi-weighted graph $(\St,E,(\upkappa_a)_{a \in A}, (\uprho_a)_{a \in A})$ such that $(s,s') \in E$ iff there exists $\delta$ such that $\tr(s,\dec) = s'$ and $(s,\dec)$ is $\tup{z}$-secure for all $a \in \Ag$.
Observe that $(s_i, \delta_i)$ being $z_a$-secure for $a$, for every $a \in \Ag$ and $i \in \SetN$, is equivalent to the fact 
that $\pi$ is contained in $G[\tup{z}]$.
Therefore, we reduced the problem of deciding whether $\SEPNE(G) \neq \emptyset$ to deciding whether there exists an ultimately periodic path $\pi \in G[\tup{z}]$ such that $z_a \lex \pay_a(\pi) + \epsilon$, for every $a \in \Ag$. We solve this problem by guessing the top priorities visited by each player, and then employing the linear programming approach described in Theorem~\ref{thm:pathfind} on the multi-weighted graph $G[\tup{z}]$ and vector of payoffs $\tup{z} + \epsilon =(z_{a_1} + \epsilon, \ldots, z_{a_n} + \epsilon)$.
In terms of complexity, we need to iterate through all possible vectors of punishment values, of which there are ${n^\prime}^{\abs{\Ag}}$, iterate through all possible priorities for each player, of which there are $\abs{\rho}^{\abs{\Ag}}$, then compute the equivalence relation recursively by solving the corresponding linear programs.
Letting $q^\prime$ be some appropriate polynomial, capturing the complexity of computing the equivalence relation, we see that this can all be done in time $O\left({n^\prime}^{\abs{\Ag}} \cdot \abs{\rho}^{\abs{\Ag}} \cdot q^\prime(n^\prime)\right)$.
Expanding this out, and letting $q^{\prime\prime}$ be appropriately chosen polynomials, we can conclude that our path finding algorithm can be done in time, 
\begin{equation*}
	O\left(2^{\abs{\Ag} \cdot 2^{q^{\prime\prime}\left(\abs{\gamma}\right)} }\right)
\end{equation*}

Thus, in total, our algorithm to determine if a given game has a finite state strict $\epsilon$ Nash equilibrium can be done in time linear in the number of weights, exponential in the number of agents and doubly exponential in the size of the goals. As such, our algorithm lies in 2EXPTIME.

With this, we can now prove the main result of this work:
\begin{theorem}
	There is a 2\exptime algorithm that, given a rational $\epsilon \geq 0$, a $\LEX(\LTL,\mp)$ game $G = \tpl{A, (\upkappa_{a})_{a \in \Ag}, \AP, \lab, (\upgamma_{a})_{a \in \Ag}}$ on arena $A = \tpl{\Ag,  \Act, \St, \iota, \tr}$ and an \LTL-formula $\Phi$, decides whether there is $\vec{\upsigma} \in \SEPNE(G)$ such that $\pi_{\vec{\upsigma}} \models \Phi$.
\end{theorem}

\begin{proof}
	
	We show that the problem can be reduced to deciding whether there exists a Nash Equilibrium in a game $G'$ defined over the arena $A' = \tpl{\Ag \cup \{a_1, a_2\},  \Act, \St \times \{0, 1\} , (\iota, 0), \tr'}$ in which $a_{1}, a_{2}$ are two new fresh agents and the transition function $\tr'$ is defined as follows.
	
	$$
	\tr'((s, \varpi), \delta) = \begin{cases}
	(\tr(s, \delta_{\rst \Ag}), 0) & \text{if } \delta(a_{1}) = \delta(a_2) \\
	(\tr(s, \delta_{\rst \Ag}), 1) & \text{otherwise}
	\end{cases}.
	$$
	
	Then, consider a fresh atomic proposition $p$ and define
	$$G' = \tpl{A', (\upkappa_{a}')_{a \in \Ag'}, \AP \cup \{p\}, \lab', (\upgamma_{a}')_{a \in \Ag'}}$$
	such that $\upkappa_{a}' = \upkappa_{a}$ for every $a \in \Ag$ and
	$\upkappa_{a_1}'(c) = \upkappa_{a_2}'(c) = 0$ for every action $c \in \Act$, and $\lab'(s,0) = \lab(s)$ and $\lab'(s,1) = \lab(s) \cup \{p\}$, for every state $s \in \St$.
	Finally, define $\gamma_{a}' = \gamma_{a}$, for every $a \in \Ag$ and $\gamma_{a_{1}}' = \Phi \vee \nextX p$ and $\gamma_{a_{2}}' = \Phi \vee \nextX \neg p$.
	
	Intuitively, the game $G'$ results from pairing $G$ with a two-player game played by agents $a_{1}$ and $a_{2}$ that are triggered to play against each other in case the formula $\Phi$ is not satisfied along the path.
	
	Now, on one hand, let $\vec{\upsigma} \in \SEPNE(G)$ such that $\pi_{\vec{\upsigma}} \models \Phi$ and consider a strategy profile $\vec{\upsigma}'$ in $G'$ such that $\vec{\upsigma}'_{\rst \Ag} = \vec{\upsigma}$~\footnote{This is an abuse of notation, as the transition functions of the two strategies are defined on a different set of decisions.
	Here we mean that the transition functions of $\vec{\upsigma}'$ simply copy the ones in $\vec{\upsigma}$ by ignoring the components expressed by agents $a_{1}$ and $a_{2}$}.
	Clearly, each agent $a \in \Ag$ takes the same sequence of actions in both strategy profiles, implying that $\mp_{a}(\pi_{\vec{\upsigma}}) = \mp_{a}(\pi_{\vec{\upsigma}'})$.
	Moreover, it is easy to see that $\lab'(\pi_{\vec{\upsigma}'})_{\rst \AP} = \lab(\pi_{\vec{\upsigma}})$~\footnote{With another abuse of notation, we here mean the restriction of sequences in $\AP \cup \{p\}$ to sequences in $\AP$.}, and so that $\pi_{\vec{\upsigma}'} \models \gamma_{a}$ if and only if $\pi_{\vec{\upsigma}} \models \gamma_{a}$ for every $a \in \Ag$.
	This means that $\sat_{a}(\pi_{\vec{\upsigma}'}) = \sat_{a}(\pi_{\vec{\upsigma}})$ and so that $\pay_{a}(\pi_{\vec{\upsigma}'}) = \pay_{a}(\pi_{\vec{\upsigma}})$.
	Moreover, since $\pi_{\vec{\upsigma}} \models \Phi$, it holds that $\pi_{\vec{\upsigma}'} \models \gamma_{a_{1}}$ and $\pi_{\vec{\upsigma}'} \models \gamma_{a_{2}}$ and that $\pay_{a_{1}}(\pi_{\vec{\upsigma}'}) = \pay_{a_{2}}(\pi_{\vec{\upsigma}'}) = (\top, 0)$.
	
	Now, assume by contradiction that $\vec{\upsigma}' \notin \SEPNE(G')$.
	First observe that, $(\top, 0)$ is the maximum payoff that players $a_{1}$ and $a_{2}$ can achieve in $G'$ and so neither of them has an incentive to deviate from $\vec{\upsigma}'$.
	Then, assume there is an agent $a \in \Ag$ and a strategy $\sigma_{a}''$ such that $\pay_{a}(\pi_{\vec{\upsigma}'}) + \epsilon \lexeq \pay_{a}(\pi_{(\vec{\upsigma}_{-a}^\prime, \vec{\upsigma}''_{a})})$.
	Then, we would also have that $\pay_{a}(\pi_{\vec{\upsigma}}) + \epsilon \lexeq \pay_{a}(\pi_{(\vec{\upsigma}_{-a}^\prime, \vec{\upsigma}''_{a})})$, in contradiction with the fact that $\vec{\upsigma} \in \SEPNE(G)$.
	
	On the other hand, assume that $\vec{\upsigma}' \in \SEPNE(G')$.
	Then, by a symmetrical reasoning, we obtain that $\vec{\upsigma} = \vec{\upsigma}'_{\rst \Ag} \in \SEPNE(G)$.
	Moreover, note that $\pi_{\vec{\upsigma}'} \models \Phi$, otherwise, either $\pay_{a_{1}}(\pi_{\vec{\upsigma}'}) = (\bot, 0)$ or $\pay_{a_{2}}(\pi_{\vec{\upsigma}'}) = (\bot, 0)$, and there would exist a beneficial deviation for one of them, contradicting the fact that $\vec{\upsigma}' \in \SEPNE(G')$.
	Hence, from $\pi_{\vec{\upsigma}'} \models \Phi$ we obtain that $\pi_{\vec{\upsigma}} \models \Phi$, which concludes the proof.	\qed
\end{proof}

\begin{theorem}
	Deciding whether there exists a finite-state strict $\epsilon$ Nash Equilibrium in a given \LEX(\LTL,\mp) game is 2\exptimeH.
\end{theorem}

\begin{proof}
	We show a reduction from the problem of finding 
	finite-state Nash Equilibria in (simple) \LTL games, whose complexity is 2\exptimeC (Cfr., see~\cite{GHW15}).
	For an $\LTL$ game $G = (A, \AP, \lab, (\gamma_{a})_{a \in \Ag})$ on the arena $A = \tpl{\Ag,  \Act, \St, \iota, \tr}$, consider the \LEX(\LTL,\mp) game $G' = (A, (\kappa_{a})_{a \in \Ag}, \AP, \lab, (\gamma_{a})_{a \in \Ag})$ over the same arena and such that $\kappa_{a}(s) = 0$ for every $s \in \St$.
	Intuitively, $G'$ is the same as $G$ but with a vacuous weighting added to match the game type of \LEX(\LTL,\mp).
	In particular, note that every strategy $\sigma_{a}$ in $G$ for player $a$ is also a strategy in $G'$ for the same player, and vice-versa.
	At this point, for a non negative $\epsilon > 0$, and denoting the set of finite-state Nash Equilibria of an LTL game, $G$, by $\FNE(G)$, we claim that $\FNE(G) = \SEPNE(G')$. The proof is by double inclusion.
	
	On one hand, let $\vec{\upsigma} \in \FNE(G)$ and assume, by contradiction that $\vec{\upsigma} \notin \SEPNE(G')$.
	Then, there exists an agent $a \in \Ag$ and a strategy $\sigma_{a}'$ such that $\pay_{a}(\pi_{\vec{\upsigma}}) + \epsilon \lexeq \pay_{a}(\pi_{(\vec{\upsigma}_{-a}, \sigma_{a}')})$.
	Now, observe that, since the $\mp$ value in $G'$ for every player is always null, the above inequality can only apply when $\pay_{a}(\pi_{\vec{\upsigma}}) = (\bot, 0)$ and $\pay_{a}(\pi_{(\vec{\upsigma}_{-a}, \sigma_{a}')}) = (\top, 0)$.
	Hence, we obtain that $\pi_{\vec{\upsigma}} \not\models \gamma_{a}$ and $\pi_{(\vec{\upsigma}_{-a}, \sigma_{a}')} \models \gamma_{a}$ which contradicts the fact that $\vec{\upsigma} \in \FNE(G)$.
	
	On the other hand, let $\vec{\upsigma} \in \SEPNE(G')$ and assume, by contradiction, that $\vec{\upsigma} \notin \FNE(G)$.
	Then, there exists an agent $a \in \Ag$ and a strategy $\sigma_{a}'$ such that $\pi_{\vec{\upsigma}} \not\models \gamma_{a}$ and $\pi_{(\vec{\upsigma}_{-a}, \sigma_{a}')} \models \gamma_{a}$.
	Hence, in $G^\prime$,
	 we have that $\pay_{a}(\pi_{\vec{\upsigma}}) + \epsilon = (\bot, \epsilon) \lexeq (\top, 0) = \pay_{a}(\pi_{(\vec{\upsigma}_{-a}, \sigma_{a}')})$, in contradiction with the fact that $\vec{\upsigma} \in \SEPNE(G')$.	\qed
\end{proof}

\section{Related Work}
\label{secn:relwork}

Our work has its origins in several threads of research in AI and mainstream computer science.
In computer science, the problem of checking whether a (typically finite state) system satisfies a specification expressed as a temporal logic formula has been a major research area since the 1970s~\cite{pnueli:77a}.
In the 1980s, the introduction of the model checking paradigm meant that verification based on temporal logic became a practical possibility, leading to a rapid growth of interest in model checking both in the verification community and beyond~\cite{clarke:2000a}.
In the late 1990s, attention began to move from the verification of \emph{closed} systems to \emph{open} systems.
A key problem in the verification of open systems is that of determining whether a particular system or system component can \emph{force} a property to hold: that is, whether there exists a strategy such that, by following that strategy, the system component can ensure that the property will be \emph{guaranteed} to hold, irrespective of the behaviour of other system components. A logic called Alternating-time Temporal Logic (ATL) was introduced to explicitly support such reasoning about strategic
ability~\cite{AHK02}.
ATL proved to be extremely influential, and was widely taken up within the multi-agent systems community.  
Although ATL embodies an important game theoretic concept---the
notion of a winning strategy---it provides no mechanism for
capturing preferences, and so its ability to capture game theoretic
concepts beyond strategic ability---such as Nash equilibrium---is
inherently limited.  For this reason, research then began to shift to
formalisms that could capture properties such as Nash equilibrium. Of
these, Strategy Logic is currently the best-known example of such a
formalism~\cite{chatterjee:2010a}.  

In the context of concurrent games and multi-agent systems, the main
decision problem in this work concerns the existence of an equilibrium
satisfying a system property; this problem is called ``rational
synthesis''~\cite{FKL10} or ``rational verification'' (cf., E-NASH
in~\cite{WGHMPT16,gutierrez:2017a}) and includes equilibrium-emptiness
as a special case. In this article, we studied this problem,
specifically, for $\LEX(\LTL,\mp)$-games and finite-state strict
$\epsilon$ Nash equilibria (via a reduction 
$\LEX(\parity,\mp)$-games), which includes, as a special case, strict Nash equilibria as well as Nash equilibria in games with \LTL goals.
All such problems can be solved in 2\textsc{ExpTime} regardless of the
particular setting.  Most other work in rational synthesis concerns ordinary (i.e., not necessarily finite-state, nor strict) NE-emptiness. In
particular, NE-emptiness has been studied for other objectives,
notably mean-payoff (\np-complete)~\cite{UW11}, B\"uchi
(\ptime-complete)~\cite{BBMU15}, and lexicographic order on B\"uchi
objectives (in \np, but not known to be \np-complete)~\cite{BBMU15}.

E-NASH (similar to what we call equilibrium-existence) for finite-state strategies has been studied on iterated
Boolean Games (a simple form of infinite-duration multiplayer
concurrent games) as follows: with LTL objectives, E-NASH is
2\textsc{ExpTime}-complete~\cite{GHW15}; with objective-LTL, {\em
  i.e.}, each agent has to optimise a reward based on the truth value
of a finite number of fixed LTL formulae, E-NASH is
2\textsc{ExpTime}-complete~\cite{KPV16}. A special case of
objective-LTL is the lexicographic order on a finite number of
components, each consisting of an \LTL formula, also
2\textsc{ExpTime}-complete.

Actually, these lower-bounds are inherited from the fact that solving two-player zero-sum games with \LTL objectives is already 2\textsc{ExpTime}-complete~\cite{RosnerThesis}.

We remark that all these works (except objective-LTL and LTL[F], discussed below) concern equilibria concepts in multiplayer games with either qualitative or quantitative objectives, but not a combination, as we do.

Most of the work that combines more than one objective has mainly been studied for the restricted setting of two-player games. There work that considers a combination of quantitative objectives such as multi-mean-payoff and multi-energy games~\cite{VCHRR15}. In the non-zero-sum case, secure-equilibria (in which each player tries to maximise their own payoff and then minimise their opponent's payoff) has been studied for a host of quantitative objectives, including mean-payoff~\cite{BruyereMR14}.

In terms of work that has studied combinations of qualitative and quantitative objectives, as we do, there are a number of results. Again, such combinations have mainly been studied in the setting of two-player games. In the zero-sum case notable works combine the parity condition with mean-payoff objectives~\cite{CHJ05,BCHJ09} (we draw on results of \cite{CHJ05} 
in our proofs) or with energy objectives~\cite{ChatterjeeD12}. On the other hand, there has been some work in the multi-agent setting, but in different settings. ~\cite{AKP18} studies the rational synthesis problem for an extension of \LTL by quality operators. Contrary to our work, such a logical extension does not account qualitative and quantitative combination of objectives, but rather aims at measuring the degree of satisfaction of a given qualitative objective. In addition, it is not expressive enough to capture mean-payoff objectives. Objective-LTL combines Boolean objectives (given as \LTL formulae) in a weak way, {\em i.e.}, there are only finitely many possible payoffs.
In contrast, $\LEX(\LTL,\mp)$ combines qualitative objectives (given as \LTL formulae) with quantitative objectives (given as mean-payoff objectives), and thus result in infinitely many possible payoffs.

Our work is somewhat related to the paradigm of \emph{Boolean games}~\cite{harrenstein:2001a,bonzon:2006a,wooldridge:2013a,grant:2011b,dunne:2008a}.
A Boolean game is a non-cooperative game played over a set of Boolean
variables. Each player in such a game desires the satisfaction of a
goal, specified as a logical formula over the overall set of
variables, and is assumed to control a subset of the variables: the
choices available to a player correspond to the assignments that can
be made to the variables controlled by that player. Players
simultaneously choose valuations for the variables they control, and a
player is satisfied if their goal is made true by the resulting
overall valuation. In addition to being an interesting game-theoretic
model in their own right, Boolean games are a natural abstract model
for studying strategic behaviour in multi-agent systems. Specifically
related to our setting is the generalisation of Boolean games known as
\emph{Iterated Boolean Games}, in which goals are specified as \LTL\
formulae, and the game takes place over an infinite number of
rounds~\cite{GHW15}.  As with our approach, most work in Boolean games
assumes pure strategies, although some attention has recently
been given to mixed strategies~\cite{ianovski:2018a}. 

Finally, it is also worth mentioning work on multi-agent planning
models such as the multi-agent STRIPS
model~\cite{brafman:2009a,brafman:2013a}. These frameworks, model systems where each agent
is an individual planning system, attempting to achieve an individual
goal. The relationship between these frameworks and concurrent games
was investigated by Gutierrez \emph{et al.}~\cite{gutierrez:2017a}.

\section{Conclusion}
\label{secn:conc}
In the last twenty years, significant effort was devoted to
analyze qualitative aspects of multi-agent systems, and more recently, also quantitative aspects. However, the two settings are often 
investigated independently.  As this is not appropriate in many
natural scenarios, researchers began to look at the combination
of these two worlds. The achievements to date in this direction,
however, are far from satisfactory, either because the settings are
too weak, ({\em e.g.}, they cannot model important solution concepts
such as Nash Equilibrium~\cite{BG13}), or because they are too
expensive in terms of complexity, ({\em e.g.}, between
\textsc{ExpTime}-hard and undecidable~\cite{ALNR15}).

In this paper we introduce a model of multi-agent systems in which
each agent's payoff is a lexicographic combination of qualitative
(\LTL) and quantitative (mean-payoff) payoffs.  We call these
$\LEX(\LTL,\mp)$ games. The solution concept we focus on is
finite-state strict $\epsilon$ Nash equilibria (for
$\epsilon \geq 0$). In this setting, we proved that the
rational synthesis problem (a generalisation of the equilibrium
existence problem) is decidable, and moreover is in \np when the 
qualitative goals are represented as parity conditions and is 
2\textsc{ExpTime}-complete when they are given by \LTL formulae.  
The proof characterises the equilibrium executions as certain 
ultimately periodic paths in a multi-weighted graph. To compute 
this graph we solve two-player zero-sum games with lexicographic 
objectives, and to find paths in such graphs we use Linear Programming.
The question
as to whether the main result also holds for ordinary equilibria ({\em i.e.}, non-strict) in games with \LTL goals has been left open. 
Indeed, our choice of using strict equilibria allowed us to supply the characterisation of such equilibria in Proposition~\ref{prop:NEchar}, and a first-step towards handling ordinary equilibria would be to supply an analogous characterisation.
Another interesting question is about the strategic ability of the agents.
Here, we analysed the case of finite-state strategies: this is motivated by the intention to provide implementable solutions for AI and multi-agent systems communities.
However, from the theoretical perspective it is interesting to consider the infinite-state case.
We can already deduce that this case is not equivalent to the finite-state one, as our setting includes the one of two-player parity mean-payoff games, for which such an inequivalence is proved in~\cite{CHJ05}.
In addition to this, a new technique would be required for addressing that case.
As a matter of fact, infinite-state strategies do not guarantee an ultimately periodic path outcome.
Therefore, the multicycle approach misses the non-periodic solutions.

For future work we foresee a number of possibilities. Some immediate
questions include investigating whether our results extend to
non-strict Nash equilibria (the question we left open) and whether
they still hold in variations of the game, for instance, in games with
more complex preference relations or in games played in different
arenas where the main complexity results may change.  Given the
2\textsc{ExpTime} nature of the main problem, we can also investigate
to what extent the main problem becomes easier in restricted
settings. For instance, typical cases of study include games with
memoryless strategies or with simpler classes of goals, such as
properties that can be described in fragments of \LTL. 
Many questions can also be asked with respect to the quantitative component of our games, {\em e.g.}, whether there is a Nash equilibrium in which each player, or some designated set of players more generally, can be ensured a mean-payoff within a certain interval. All of these
questions constitute important avenues for further work.

\section{Acknowledgments}
The authors would like to thank the reviewers for their careful reading of preliminary versions of this manuscript.
This allowed us to significantly improve the quality and solidity of our results.

\bibliographystyle{abbrv}
\bibliography{References}

\end{document}